\documentclass[copyright,creativecommons]{eptcs}
\providecommand{\event}{AiML 2026} 

\usepackage{aiml26}
\usepackage{iftex}
\usepackage{amssymb}
\usepackage{empheq}
\usepackage{xcolor}

\ifpdf
  \usepackage{underscore}
  \usepackage[T1]{fontenc}
\else
  \usepackage{breakurl}
\fi

\title{Knowing-Value Logic with Successor Arithmetic}
\author{Hongyi Wang
\institute{Department of Philosophy\\
Peking University\\
Beijing, China}
\email{wanghongyi@stu.pku.edu.cn}
}

\newcommand{\titlerunning}{Knowing-Value Logic with Successor Arithmetic}
\newcommand{\authorrunning}{Hongyi Wang}

\newcommand{\K}{K}
\newcommand{\Khat}{\widehat{K}}
\newcommand{\Kv}{Kv}
\newcommand{\Ag}{\mathrm{Ag}}
\newcommand{\dep}{\mathrm{depth}}
\newcommand{\hgt}{\mathrm{height}}
\newcommand{\Last}{\mathrm{Last}}

\newcommand{\logic}{\mathbf{ELKvSA}^r}
\newcommand{\sys}{\mathbb{ELKVSA}^r}
\hypersetup{
  bookmarksnumbered,
  pdftitle    = {\titlerunning},
  pdfauthor   = {\authorrunning},
  pdfsubject  = {Epistemic Logic and Arithmetic},
  pdfkeywords = {knowing-value logic, successor arithmetic, epistemic logic, completeness}
}

\begin{document}
\maketitle

\begin{abstract}
In their prior work~\cite{WangFan2014Conditionally}, Wang and Fan proposed conditional knowing-value logic and provided a complete axiomatization. However, in natural language scenarios and logic puzzles, knowing-value reasoning often appears together with arithmetic operations, which motivates us to enrich knowing-value logic with arithmetic function symbols. In this paper, we extend the language of conditional knowing-value logic with equality and the successor function. Due to the failure of compactness over the class of standard models, we additionally introduce non-standard models to facilitate the technical analysis. Our main results establish the finite model property and provide an axiomatization that is strongly complete with respect to the class of non-standard models and weakly complete with respect to the class of standard models. Furthermore, we extend our logic with public announcement operators and use the resulting system to formalize and solve the ``Consecutive Numbers'' puzzle. This work provides a novel framework for integrating epistemic logic with arithmetic.
\end{abstract}
\section{Introduction}
Epistemic logic is a branch of philosophical logic that studies reasoning about knowledge, belief, and related concepts. In this field, propositional knowledge---usually expressed as ``knowing-that''---has been widely studied in classical epistemic logic. However, as \cite{wang2018beyond} points out, natural language also contains many ``knowing-wh'' expressions, such as knowing-what, knowing-why, knowing-how, and knowing-who, which require systematic logical analysis. Therefore, logics of ``knowing-wh'' have attracted much attention in recent years.

This paper focuses on ``knowing-value'' logic, a specific subfield of ``knowing-wh'' research. It formalizes expressions such as ``S knows the value of the password'' or ``S knows what the password is''. In \cite{Plaza1989Logics}, Plaza first proposed the knowing-value operator $\Kv$. Intuitively, $\Kv_i c$ means that agent $i$ knows the value of the non-rigid constant $c$. Plaza also provided an axiomatization for $\mathbf{ELKv}$ (epistemic logic with the knowing-value operator). Building on this work, Wang and Fan introduced a relativized knowing-value operator in \cite{Wang2013Knowingthat}. Intuitively, $\Kv_i(\phi,c)$ means that agent $i$ knows the value of $c$ given the condition $\phi$. This operator can formalize statements such as ``S knows the value of the password, given that the password is a 4-digit number''. In \cite{Wang2013Knowingthat,WangFan2014Conditionally}, they provided a complete axiomatization for the corresponding logics $\mathbf{ELKv}^r$ and $\mathbf{PALKv}^r$ (public announcement logic with the relativized knowing-value operator). There are also other studies on knowing-value logic, such as \cite{Hong2023Knowing}, which introduced the $\K vP$ operator to express knowing the value of a predicate.

However, the logics $\mathbf{ELKv}^r$ and $\mathbf{PALKv}^r$ still have some limitations in both syntax and semantics. 
From a syntactic perspective, because they lack equality and numerical terms, their formalization can be coarse-grained. For instance, $\mathbf{ELKv}^r$ expresses ``S knows the value of the password given that it is a 4-digit number'' as $\Kv(p,c)$, using a propositional letter $p$ for the condition. Similarly, ``S knows that if the password is 4-digit, then it is 1111'' is formalized as $\K(p\to q)$, using $q$ for ``the password is 1111''. This approach hides the internal structure of these propositions. As a result, $\mathbf{ELKv}^r$ fails to capture some reasoning in natural language. In natural language, the second statement clearly implies the first. But in $\mathbf{ELKv}^r$, $\K(p\to q)$ does not imply $\Kv(p,c)$ because $p$ and $q$ are just independent letters. 
Generally, reasoning that involves the interaction between equality, knowing-value, and arithmetic is common in natural language scenarios, but $\mathbf{ELKv}^r$ cannot fully express it. 
From a semantic perspective, a model for $\mathbf{ELKv}^r$ assigns values to constants over a bare, unstructured set of objects $O$. To study specific mathematical structures, such as natural numbers with the successor function $(\mathbb{N},0^\mathbb{N}, S^\mathbb{N})$ or with addition and multiplication $(\mathbb{N}, 0^\mathbb{N},1^\mathbb{N}, +^\mathbb{N}, \times^\mathbb{N})$, we must enrich the language with arithmetic function symbols and add arithmetic axioms to the proof system. 
All these reasons motivate us to enrich knowing-value logic with equality and arithmetic function symbols, starting with the successor function. To further demonstrate the need for this extension, we present the following puzzle.
\begin{example}[\cite{van2015one}]
\label{ex:consecutive-puzzle}
		Anne and Bill get to hear the following: ``Given are two natural numbers. They
	are consecutive numbers. I am going to whisper one of these numbers to Anne and
	the other number to Bill.'' This happens. Anne and Bill now have the following
	conversation.
	\begin{itemize}
		\item Anne: ``I don't know your number.''
		\item Bill: ``I don't know your number.''
		\item Anne: ``I know your number.''
		\item Bill: ``I know your number.''
	\end{itemize}
	How is that possible? What surely is one of the two numbers?
\end{example}

In this puzzle, introducing the successor function allows us to formalize the statement ``Anne knows that $a$ and $b$ are consecutive numbers, but Anne does not know the value of $b$'' as $K_A (a \approx Sb \vee b \approx Sa) \wedge \neg Kv_A b$. Instead of using a simple propositional letter $p$ to represent ``$a$ and $b$ are consecutive'', this formalization preserves the internal structure and relevant details of the puzzle. 
Furthermore, the reasoning in the puzzle crucially relies on the premise that Anne and Bill both know that $0$ is not a successor in natural numbers. In first-order logic, this can be derived by an axiom of successor arithmetic: $\forall x \neg (Sx \approx 0)$. Although our logic does not have quantifiers, we can capture this by introducing an axiom schema $\neg Sc \approx 0$ for any constant $c$. We formalize and solve the puzzle in Section~\ref{sec:consecutive-numbers}.

Our main contributions are as follows. First, we introduce $\mathbf{PALKvSA}^r$, an extension of $\mathbf{PALKv}^r$ with equality and the successor function; its announcement-free fragment is $\logic$. We define semantics for these languages over both standard and non-standard models. Second, for $\logic$, we propose the axiom system $\sys$ and prove that it is strongly complete with respect to non-standard models and weakly complete with respect to standard models. Third, we establish that $\logic$ has the finite model property and is decidable. Finally, we show how the ``Consecutive Numbers'' puzzle can be formalized and solved within this setting.
Overall, our work develops a novel framework that integrates epistemic logic with arithmetic, which is of independent theoretical interest and has many potential applications.

The rest of the paper is organized as follows. Section~\ref{sec:preliminaries} briefly reviews the preliminaries. Section~\ref{sec:language-semantics} introduces the language and semantics of our logic. Section~\ref{sec:axiomatization} presents the axiom system $\sys$ and proves its strong completeness over non-standard models. Section~\ref{sec:standard-models} establishes the finite model property, weak completeness over standard models, and decidability. Section~\ref{sec:consecutive-numbers} formalizes the puzzle using public announcement operators. Finally, Section~\ref{sec:conclusion} concludes the paper.

\section{Preliminaries}
\label{sec:preliminaries}
This section briefly recalls the background needed in our work: the conditional knowing-value logic and first-order successor arithmetic.
\subsection{Conditional Kv Logic}
\begin{definition}[Language $\mathbf{PALKv}^r$]
	Given a countably infinite set of proposition letters $\mathbf{P}$, a countably infinite set
of agent names $\mathbf{I}$, and a countably infinite set of (non-rigid) constant symbols
$\mathbf{C}$, the language of $\mathbf{PALKv}^r$ is defined as follows:
	\[
\phi ::= \top \mid p \mid \neg \phi \mid \phi \land \phi \mid K_i \phi \mid Kv_i(\phi, c) \mid [\phi] \phi
\]
where $p \in \mathbf{P}, i \in \mathbf{I}, c \in \mathbf{C}$.
\end{definition}
We use the standard abbreviations $\phi\to\psi:=\neg(\phi\land\neg\psi)$, $\phi\lor\psi:=\neg(\neg\phi\land\neg\psi)$, $\phi\leftrightarrow\psi:=(\phi\to\psi)\land(\psi\to\phi)$, $\bot:=\neg\top$, and $\widehat{K}_i\phi:=\neg K_i\neg\phi$.
	A model $\mathcal{M}$ is defined as $\langle W,\{\sim_i\mid i\in I\}, O,V, V_\mathbf{C}\rangle$, where $W$ is a non-empty set of possible worlds, $\sim_i$ is an equivalence relation on $W$, $O$ is a non-empty set of objects, $V$ is a valuation function assigning a set of worlds to each $p\in\mathbf{P}$, and $V_\mathbf{C}:\mathbf{C}\times W\rightarrow O$ assigns an object for each constant-world pair. For a formula $\psi$, write $\mathcal{M}|_\psi$ for the restriction of $\mathcal{M}$ to $W_\psi=\{w\in W\mid \mathcal{M},w\vDash\psi\}$. The semantics is defined as follows:
\[
\begin{array}{|lcl|}
\hline
\mathcal{M}, w \vDash \top & & \text{always holds} \\
\mathcal{M}, w \vDash p & \Leftrightarrow & w \in V(p) \\
\mathcal{M}, w \vDash \neg \phi & \Leftrightarrow & \mathcal{M}, w \not\vDash \phi \\
\mathcal{M}, w \vDash \phi \land \psi & \Leftrightarrow & \mathcal{M}, w \vDash \phi \text{ and } \mathcal{M}, w \vDash \psi \\
\mathcal{M}, w \vDash K_i \psi & \Leftrightarrow & \text{for all } u \text{ such that } w \sim_i u: \mathcal{M}, u \vDash \psi \\
\mathcal{M}, w \vDash Kv_i(\phi, c) & \Leftrightarrow & \text{for any } u, v \in W \text{ such that } w \sim_i u \text{ and } w \sim_i v: \\
& & \quad \text{if } \mathcal{M}, u \vDash \phi \text{ and } \mathcal{M}, v \vDash \phi, \text{ then } V_{\mathbf{C}}(c, u) = V_{\mathbf{C}}(c, v) \\
\mathcal{M}, w \vDash [\psi] \phi & \Leftrightarrow & \mathcal{M}, w \vDash \psi \text{ implies } \mathcal{M}|_\psi, w \vDash \phi\\
\hline
\end{array}
\]
We denote the announcement-free part of $\mathbf{PALKv}^r$ as $\mathbf{ELKv}^r$.
It is shown in \cite{WangFan2014Conditionally} that the following system $\mathbb{ELKV}^r$ completely axiomatizes $\mathbf{ELKv}^r$.
\begin{center}
$\text{System } \mathbb{ELKV}^r$
\end{center}
\begin{tabular}{cc}
\begin{tabular}[t]{lc}
\textbf{Axiom Schemas} & \\
TAUT & all the instances of tautologies \\
DISTK & $K_i(\phi \to \psi) \to (K_i\phi \to K_i\psi)$ \\
T & $K_i\phi \to \phi$ \\
4 & $K_i\phi \to K_iK_i\phi$ \\
5 & $\neg K_i\phi \to K_i\neg K_i\phi$ \\
DISTKv$^r$ & $K_i(\phi \to \psi) \to (Kv_i(\psi, c) \to Kv_i(\phi, c))$ \\
Kv$^r$4 & $Kv_i(\phi, c) \to K_iKv_i(\phi, c)$ \\
Kv$^r\bot$ & $Kv_i(\bot, c)$ \\
Kv$^r\lor$ & $\widehat{K}_i(\phi \land \psi) \land Kv_i(\phi, c) \land Kv_i(\psi, c) \to Kv_i(\phi \lor \psi, c)$
\end{tabular}
&
\begin{tabular}[t]{lc}
\textbf{Rules} & \\
MP & $\dfrac{\phi, \phi \to \psi}{\psi}$ \\
NECK & $\dfrac{\phi}{K_i\phi}$ \\
RE & $\dfrac{\psi \leftrightarrow \chi}{\phi \leftrightarrow \phi[\psi/\chi]}$
\end{tabular}
\end{tabular}
\medskip

\noindent By reduction, $\mathbb{ELKV}^r$ together with the following reduction axioms completely axiomatize $\mathbf{PALKv}^r$. The resulting system is called $\mathbb{PALKV}^r$.
\[
{
\begin{array}{l l}
\text{!ATOM} & [\psi]p \leftrightarrow (\psi \to p) \\
\text{!NEG}  & [\psi]\neg\phi \leftrightarrow (\psi \to \neg[\psi]\phi) \\
\text{!CON}  & [\psi](\phi \land \chi) \leftrightarrow ([\psi]\phi \land [\psi]\chi) \\
\text{!K}     & [\psi]\K_i \phi \leftrightarrow (\psi \to \K_i[\psi]\phi)\\
\text{!Kv}^r &  [\phi] \Kv_i(\psi, c) \leftrightarrow \bigl( \phi \to \Kv_i\bigl(\phi\wedge[\phi] \psi, c\bigr)\bigr)\\
\end{array}
}
\]

\noindent The following reduction axiom !COM is derivable from $\mathbb{PALKV}^r$.
\[
\text{!COM} \quad [\psi][\chi]\phi \leftrightarrow [\psi \land [\psi]\chi]\phi
\]

\subsection{\texorpdfstring{First-Order Successor Arithmetic}{First-Order Successor Arithmetic}}
\begin{theorem}
	Given the symbol set $\{0, S\}$ and the structure $\mathfrak{N}=(\mathbb{N},\mathbf{S})$ where $\mathbb{N}$ is the set of natural numbers and $\mathbf{S}$ is the successor function on it, the first-order theory of this structure $Th(\mathfrak{N})$ is axiomatized by the following three axioms and an additional axiom schema.
	\begin{align*}
		&\forall x(\neg Sx \approx 0)\\
		&\forall x\forall y (Sx\approx Sy\rightarrow x\approx y)\\
		&\forall x (\neg x\approx 0\rightarrow \exists y x\approx Sy)\\
		&\forall x (\neg S^n(x)\approx x) \text{ for every natural number }n\geq 1
		\end{align*}
	where $S^n(x)$ is an abbreviation of $\underbrace{SS\cdots S}_{n\text{ times}}(x)$.
\end{theorem}
\begin{proof}
	Its proof can be found in many first-order logic textbooks, for example, \cite{EbbinghausMathematicalLogic}.
\end{proof}
\section{Language and Semantics}
\label{sec:language-semantics}

In this section, we formally define the language and semantics of our logic. We define both standard and non-standard models. Note that the classes of standard and non-standard models induce two distinct semantics. When a property or result is stated w.r.t.\ standard (or non-standard) models, it is to be understood as holding under the semantics induced by the respective class of models.
We now introduce the language $\mathbf{PALKvSA}^r$ and write $\mathbf{ELKvSA}^r$ for its announcement-free fragment. 
\begin{definition}[Language $\mathbf{PALKvSA}^r(\mathbf{I},\mathbf{C})$]
    	Given a countable constant symbol set $\mathbf{C}$ and a countable set of agents $\mathbf{I}$, the set of terms $\mathrm{Term}$ and formulas are defined as follows.
	\begin{align*}
		t&::= 0\mid c\mid St\\
		\varphi &::=\top\mid t\approx t\mid\neg\varphi\mid \varphi\wedge\varphi\mid \K_i\varphi \mid \Kv_i(\varphi,t)\mid [\varphi]\varphi
	\end{align*}
	where $c\in \mathbf{C}$, $i\in \mathbf{I}$, $0$ is a rigid constant symbol, $S$ is the successor function symbol. Formulas of the form $t\approx t'$ are called atomic formulas (in $\mathbf{PALKv}^r$, atomic formulas are propositional letters). We write $Kv_i t$ as an abbreviation of $Kv_i(\top,t)$.
 \end{definition}
    \begin{definition}[Semantics]
	A standard model is a tuple $\mathcal{M} = \langle W,\mathbb{N}, \mathbf{S}, \{\sim_i\}_{i\in \mathbf{I}}, V_\mathbf{C}\rangle$, where $\sim_i$ is an equivalence relation on the set of possible worlds $W$ for each agent $i$. Here $\mathbb{N}$ is the set of natural numbers and $\mathbf{S}$ is the successor function on $\mathbb{N}$, and $V_\mathbf{C}: \mathbf{C}\times W \rightarrow \mathbb{N}$ assigns a natural number for each constant-world pair.
\end{definition}
\noindent Define the assignment for terms $V_\mathrm{T}: \mathrm{Term} \times W\rightarrow \mathbb{N}$ as follows:
\begin{align*}
	V_\mathrm{T}(0,w) &= \mathbf{0}\\
	V_\mathrm{T}(c,w) &= V_\mathbf{C}(c,w)\\
	V_\mathrm{T}(St,w) &= \mathbf{S}(V_\mathrm{T}(t,w))
\end{align*}
For a formula $\psi$, write $\mathcal{M}|_\psi$ for the restriction of $\mathcal{M}$ to $W_\psi=\{w\in W\mid \mathcal{M},w\vDash\psi\}$.
Then the satisfaction relation is defined as follows:
\begin{center}
    \fbox{
        \begin{tabular}{l c l}
            $\mathcal{M},w \models t_1 \approx t_2$ & $\Leftrightarrow$ & $V_\mathrm{T}(t_1,w)=V_\mathrm{T}(t_2,w)$ \\[1ex]
            $\mathcal{M},w \models \Kv_i(\varphi,t)$ & $\Leftrightarrow$ & for any $u,v\in W$ such that $w\sim_i u, w\sim_i v$, \\
            & & if $\mathcal{M},u \models \varphi$ and $\mathcal{M},v \models \varphi$, then $V_\mathrm{T}(t,u)=V_\mathrm{T}(t,v)$\\
            $\mathcal{M},w \models [\varphi]\psi$ & $\Leftrightarrow$ & $\mathcal{M},w \models \varphi$ implies $\mathcal{M}|_\varphi,w \models \psi$
        \end{tabular}
    }
\end{center}
\medskip

Similar to $\mathbf{PALKv}^r$, if an axiom system $\mathbb{S}$ completely axiomatizes $\mathbf{ELKvSA}^r$, then $\mathbb{S}$ combined with the reduction axioms !ATOM, !NEG, !CON, !K, $!\text{Kv}^r$ and RE completely axiomatizes $\mathbf{PALKvSA}^r$, where !ATOM is $[\varphi]\alpha \leftrightarrow (\varphi \to \alpha)$ for any atomic formula $\alpha$.
Therefore, we will focus exclusively on axiomatizing $\mathbf{ELKvSA}^r$ in the remainder of this paper.
However, as the following proposition demonstrates, $\mathbf{ELKvSA}^r$ is not compact with respect to the class of standard models. Consequently, there is no sound and strongly complete axiomatization\footnote{Strong completeness means that, for every set of formulas $\Gamma$ and formula $\varphi$, if $\Gamma\models \varphi$, then $\Gamma\vdash \varphi$; Weak completeness means that, for every formula $\varphi$, if $\models \varphi$, then $\vdash \varphi$.} for $\mathbf{ELKvSA}^r$ over standard models. 
\begin{proposition}
	\label{prop:noncompact}
	$\mathbf{ELKvSA}^r$ is not compact w.r.t. standard models.
\end{proposition}
\begin{proof}
	Consider the formula set $\{\neg c\approx S^n 0\mid n\in\mathbb{N}\}$. It is finitely satisfiable, but not satisfiable. Also consider the formula set $\{\Kv_ic\}\cup \{\neg \K_i(c\approx S^n0)\mid n\in\mathbb{N}\}$. It is also finitely satisfiable, but not satisfiable.
\end{proof}
\noindent Therefore, for technical reasons, we introduce non-standard models to facilitate our analysis.
$(\mathbb{N}^\star, \mathbf{S}^\star)$ is a non-standard model of first-order successor arithmetic, where $\mathbb{N}^\star= \mathbb{N}\cup\{(n,z)\mid n\in \mathbb{N}, z\in \mathbb{Z}\}$, and $\mathbf{S}^\star(n)=n+1$, $\mathbf{S}^\star(n,z)=(n,z+1)$.
A non-standard model for our logic is then a tuple based on this arithmetic structure:
     $$\mathcal{N} = \langle W,\mathbb{N}^\star, \mathbf{S}^\star, \{\sim_i\}_{i\in I}, V_\mathbf{C}^\star\rangle$$
The definition of $V_\mathbf{C}^\star$, $V_\mathrm{T}^\star$ and the satisfaction relation are as before, with $(\mathbb{N}^\star,\mathbf{S}^\star)$ in place of $(\mathbb{N},\mathbf{S})$. For convenience, let $Z_n:=\{(n,z)\mid z\in \mathbb{Z}\}$ be the $n$th $\mathbb{Z}$-chain. 

\begin{remark}
    Note that our non-standard models are defined to have countably infinitely many $\mathbb{Z}$-chains (indexed by $n \in \mathbb{N}$). This specific cardinality is intentionally chosen. On one hand, if we restrict the class to non-standard models with only finitely many $\mathbb{Z}$-chains, compactness would still fail. On the other hand, introducing uncountably many $\mathbb{Z}$-chains is unnecessary, as constant symbols are countable.
In the rest of the paper, ``non-standard models'' means models based on the fixed arithmetic structure $(\mathbb{N}^\star,\mathbf{S}^\star)$ above. However, once completeness is proved for this class, it also holds for any larger class of models that contains it, such as the class obtained by allowing arbitrary non-standard models of successor arithmetic.
\end{remark}

\section{Axiomatization of \texorpdfstring{$\logic$}{ELKvSA\string^r}}
\label{sec:axiomatization}
In this section, we present the axiom system $\sys$ for the logic $\logic$. We first establish the soundness of the system with respect to both standard and non-standard models. Subsequently, we prove the strong completeness of $\sys$ over the class of non-standard models.

\subsection{\texorpdfstring{System $\sys$}{System ELKVSA\string^r}}
{\small
\begin{center}
\begin{tabular}{cc}
\multicolumn{2}{c}{$\text{System } \mathbb{ELKVSA}^{r}$}\\[0.3em]
\begin{tabular}[t]{lc}
\textbf{Axiom Schemas} & \\
TAUT & all the instances of tautologies \\
DISTK & $\K_i(\varphi\to\psi)\to(\K_i\varphi\to\K_i\psi)$ \\
T & $\K_i\varphi\rightarrow\varphi$ \\
4 & $\K_i\varphi\rightarrow \K_i\K_i\varphi$ \\
5 & $\neg \K_i\varphi\rightarrow \K_i\neg \K_i\varphi$ \\
ID & $t\approx t$ \\
SI & $t\approx t'\rightarrow (\alpha\rightarrow\alpha(t'/t))$ \\
& where $\alpha$ is an atomic formula \\
SA1 & $\neg 0\approx St$ \\
SA2 & $St_1\approx St_2\rightarrow t_1\approx t_2$ \\
SA3 & $\neg S^n t\approx t, \text{ where }n\ge 1$ \\
Kv$^r$0 & $\Kv_i 0$ \\
Kv$^r$S & $\Kv_i (\varphi, t)\leftrightarrow\Kv_i (\varphi, St)$ \\
Kv$^r$4 & $\Kv_i (\varphi, t)\rightarrow \K_i\Kv_i (\varphi, t)$ \\
KVE1 & $\Kv_it_1 \wedge \K_i(\varphi \to t_1 \approx t_2) \to \Kv_i(\varphi, t_2)$ \\
KVE2 & $\Kv_it_1 \wedge \Kv_it_2 \to (\widehat{\K}_i(t_1 \approx t_2) \to \K_i(t_1 \approx t_2))$
\end{tabular}
&
\begin{tabular}[t]{lc}
\textbf{Rules} & \\
MP & $\dfrac{\varphi\quad\varphi\to\psi}{\psi}$ \\
NECK & $\dfrac{\varphi}{\K_i\varphi}$ \\
RE & $\dfrac{\psi\leftrightarrow\chi}{\varphi\leftrightarrow\varphi[\psi/\chi]}$ \\
$\text{Kv}^r$-Elim & $\dfrac{(\Kv_ic \wedge \K_i(\varphi \to c \approx t)) \to \psi}{\Kv_i(\varphi, t) \to \psi}$ \\
& where $c$ does not occur in $\varphi, \psi, t$
\end{tabular}
\end{tabular}
\end{center}
}

\begin{proposition}
    The $\text{Kv}^r$-Elim rule preserves validity w.r.t. both standard and non-standard models. 
\end{proposition}
\begin{proof}
    Assume for contradiction that the premise $\models (\Kv_ic \wedge \K_i(\varphi \to c \approx t)) \to \psi$ is valid, but the conclusion $\models \Kv_i(\varphi, t) \to \psi$ is not.
    Since the conclusion is not valid, there exists a model $\mathcal{M} = \langle W,\mathbb{O}, \mathbf{S}^\mathbb{O}, \{\sim_i\mid i\in I\}, V_\mathbf{C} \rangle$ and a world $w \in W$ such that $\mathcal{M}, w \models \Kv_i(\varphi, t)$ but $\mathcal{M}, w \not\models \psi$, where $(\mathbb{O},\mathbf{S}^\mathbb{O})$ is either $(\mathbb{N}, \mathbf{S})$ or $(\mathbb{N}^\star, \mathbf{S}^\star)$.
     From $\mathcal{M}, w \models \Kv_i(\varphi, t)$, we know there exists an element $o \in \mathbb{O}$ such that for all $u \in W$ with $w \sim_i u$, if $\mathcal{M}, u \models \varphi$, then $V_\mathrm{T}(t, u) = o$.
    Now, we construct a new model $\mathcal{M}' = \langle W,\mathbb{O}, \mathbf{S}^\mathbb{O}, \{\sim_i\mid i\in I\}, V'_\mathbf{C} \rangle$ which is exactly the same as $\mathcal{M}$, except that we force the interpretation of the constant $c$ to be $o$ in all worlds (i.e., $V'_\mathbf{C}(c, v) = o$ for all $v \in W$).
    Because $c$ does not occur in $\varphi, \psi$, or $t$, the truth values of these formulas and the valuation of $t$ remain completely unchanged in $\mathcal{M}'$. In particular, we still have $\mathcal{M}', w \not\models \psi$.
    However, let us evaluate the premise in $\mathcal{M}'$ at $w$:
    \begin{itemize}
        \item Since $c$ is assigned the same value $o$ globally, $\mathcal{M}', w \models \Kv_ic$ is true.
        \item For any $u \in W$ such that $w \sim_i u$, if $\mathcal{M}', u \models \varphi$, then $\mathcal{M}, u \models \varphi$, which implies $V'_\mathrm{T}(t, u) = V_\mathrm{T}(t, u) = o$. Since $V'_\mathbf{C}(c, u) = o$, we have $V'_\mathbf{C}(c, u) = V'_\mathrm{T}(t, u)$, meaning $\mathcal{M}', u \models c \approx t$. Thus, $\mathcal{M}', w \models \K_i(\varphi \to c \approx t)$ is true.
    \end{itemize}
    Combining these, we get $\mathcal{M}', w \models \Kv_ic \wedge \K_i(\varphi \to c \approx t)$. 
    Since we assumed the premise $(\Kv_ic \wedge \K_i(\varphi \to c \approx t)) \to \psi$ is globally valid, it must hold in $\mathcal{M}'$ at $w$. This forces $\mathcal{M}', w \models \psi$.
    This directly contradicts our earlier deduction that $\mathcal{M}', w \not\models \psi$. Therefore, our initial assumption must be false, and the rule preserves validity.
\end{proof}
\begin{theorem}
    $\sys$ is sound w.r.t. both standard and non-standard models.
\end{theorem}

We now explain the intuition behind the axioms and rules of $\sys$. The axioms ID and SI govern equality; from them, we can derive the reflexivity, symmetry, and transitivity of equality. SA1-SA3 come from first-order successor arithmetic. Note that we cannot express the property ``every non-zero number is a successor'' because our language lacks quantifiers.
The axioms KVE1 and KVE2 bridge knowing-value and equality. KVE1 states that if agent $i$ knows the value of $t_1$ and knows that $\varphi$ implies $t_1 \approx t_2$, then agent $i$ also knows the value of $t_2$ given condition $\varphi$. KVE2 states that if agent $i$ knows the values of both $t_1$ and $t_2$, then the agent knows whether they are equal.

A notable feature of the system is the $\text{Kv}^r$-Elim rule. Semantically, the formula $\Kv_i(\varphi,c)$ can be understood as $\exists x \K_i(\varphi \to c \approx x)$. Conceptually, the $\text{Kv}^r$-Elim rule plays a role similar to $\exists$-Elimination in first-order logic, allowing us to safely eliminate the relativized knowing-value operator in derivations. This rule is strong enough to derive standard axioms such as $\text{DISTKv}^r$ and $\text{Kv}^r\vee$ from $\mathbb{ELKV}^r$. Furthermore, we will show that $\text{GKv}^r\vee$ and $\text{GKVE}$---the generalized versions of $\text{Kv}^r\vee$ and $\text{KVE}$---are also derivable in $\sys$ via this rule.

Before deriving theorems in the axiom system,
we provide a general proof method using the $\text{Kv}^r$-Elim rule to derive formulas of the following form:
\( \left( \bigwedge_{k=0}^n \Kv_i(\varphi_k, t_k) \right) \to \psi \).
To prove this using $\text{Kv}^r$-Elim, we only need to show the following implication for $n+1$ fresh, pairwise distinct constants $c_0, \dots, c_n$ that do not occur anywhere in the original theorem ($\varphi_k, t_k,$ or $\psi$):
\[ \left( \bigwedge_{k=0}^n \bigl( \Kv_ic_k \wedge \K_i(\varphi_k \to c_k \approx t_k) \bigr) \right) \to \psi \]
Once this is established, we can repeatedly apply $\text{Kv}^r$-Elim to eliminate the freshly introduced constants $c_k$ one by one, thereby deriving the original theorem. 

The following propositions present several theorems derivable in $\sys$ that will be used in our later completeness proof.
\begin{proposition}
    The symmetry and transitivity of equality are derivable in $\sys$. Moreover, $t_1\approx t_2\to St_1 \approx St_2$ is also derivable.
\end{proposition}
\begin{proposition}
    \label{prop:derive}
    Both $\text{GKv}^r\vee$ and $\text{GKVE}$ are derivable in the system $\mathbb{ELKVSA}^{r}$. 
    \begin{enumerate}
        \item $\text{GKv}^r\vee$: $\left( \bigwedge_{k=0}^n \Kv_i(\phi_k, t_k) \right) \wedge \left( \bigwedge_{k=0}^{n-1} \Khat_i(\phi_k \wedge \phi_{k+1} \wedge t_k \approx t_{k+1}) \right) \rightarrow \Kv_i\left( \bigvee_{k=0}^n(\phi_k \wedge t \approx t_k), t \right)$
        \item $\text{GKVE}$: $\left( \bigwedge_{k=0}^n \Kv_i(\phi_k, t_k) \right) \wedge \left( \bigwedge_{k=0}^{n-1} \Khat_i(\phi_k \wedge \phi_{k+1} \wedge t_k \approx t_{k+1}) \right) \rightarrow \K_i\left( \phi_0 \wedge \phi_n \to t_0 \approx t_n \right)$
    \end{enumerate}
\end{proposition}
\begin{proof}
    By the general proof method above, we introduce $n+1$ fresh, pairwise distinct constants $c_0, \dots, c_n$ which do not occur in $\phi_k$, $t_k$, or $t$. Let
  $\Gamma := \bigwedge_{k=0}^n \bigl( \Kv_ic_k \wedge \K_i(\phi_k \to c_k \approx t_k) \bigr)$ and 
        $\Delta := \bigwedge_{k=0}^{n-1} \Khat_i(\phi_k \wedge \phi_{k+1} \wedge t_k \approx t_{k+1})$.
To derive the two theorems, we only need to prove that $\Gamma \to (\Delta \to \psi)$ holds in $\mathbb{ELKVSA}^{r}$ for each respective consequent $\psi$. Assume $\Gamma$ and $\Delta$.
    
    First, for each $0\leq k\leq n-1$, from $\K_i(\phi_k \to c_k \approx t_k)$ and $\K_i(\phi_{k+1} \to c_{k+1} \approx t_{k+1})$, we can derive
    \( \K_i \bigl( (\phi_k \wedge \phi_{k+1} \wedge t_k \approx t_{k+1}) \to c_k \approx c_{k+1} \bigr) \).
    Combining this with $\widehat{\K}_i(\phi_k \wedge \phi_{k+1} \wedge t_k \approx t_{k+1})$ from $\Delta$ yields $\widehat{\K}_i(c_k \approx c_{k+1})$. 
    By $\text{KVE2}$ and $\Kv_i c_k$ and $\Kv_i c_{k+1}$, we derive $\K_i(c_k \approx c_{k+1})$ for each $0\le k\le n-1$. 
    By the transitivity of equality inside the $\K_i$ modality, this naturally yields $\K_i(c_0 \approx c_k)$ for all $0 \le k \le n$.
\begin{enumerate}
    \item For $\text{GKv}^r\vee$:  From $\K_i(\phi_k \to c_k \approx t_k)$ in $\Gamma$ and $\K_i(c_0 \approx c_k)$ derived above, we deduce $\K_i(\phi_k \to c_0 \approx t_k)$, which implies $\K_i\bigl( (\phi_k \wedge t \approx t_k) \to c_0 \approx t \bigr)$. Taking the disjunction of the antecedents gives:
    \( \K_i\left( \bigvee_{k=0}^n (\phi_k \wedge t \approx t_k) \to c_0 \approx t \right) \).
    Since we have $\Kv_i c_0$ from $\Gamma$, by $\text{KVE1}$, we can deduce $\Kv_i\left( \bigvee_{k=0}^n (\phi_k \wedge t \approx t_k), t \right)$. 
    \item For $\text{GKVE}$: With $\K_i(c_0 \approx c_n)$ derived, along with $\K_i(\phi_0 \to c_0 \approx t_0)$ and $\K_i(\phi_n \to c_n \approx t_n)$ from $\Gamma$, standard equality reasoning yields $\K_i(\phi_0 \wedge \phi_n \to t_0 \approx t_n)$. 
\end{enumerate}
\end{proof}

\begin{proposition}
    The following theorems can be derived in $\mathbb{ELKVSA}^{r}$:
    \begin{enumerate}
        \item $\text{DISTKv}^r: \K_i(\varphi \to \psi) \to (\Kv_i(\psi, t) \to Kv_i(\varphi, t))$
        \item $\text{Kv}^r\bot: \Kv_i(\bot, t)$
        \item $\text{Kv}^r\vee:\Khat_i(\phi \land \psi) \land Kv_i(\phi, c) \land Kv_i(\psi, c) \to Kv_i(\phi \lor \psi, c)$
        \item $\text{Kv}^r5:\neg Kv_i(\varphi, t) \to \K_i\neg Kv_i(\varphi, t)$
    \end{enumerate}
\end{proposition}
\begin{proof}
    We only show the proofs for the first two items. The proof of 3 is similar to the above proposition, and the proof of 4 can be found in \cite{WangFan2014Conditionally}.
    \begin{enumerate}
        \item We rewrite the axiom as $\Kv_i(\psi, t) \to (\K_i(\varphi \to \psi) \to \Kv_i(\varphi, t))$. Using the $\text{Kv}^r$-Elim method, we introduce a fresh constant $c$ and aim to prove:
        \( \bigl( \Kv_i c \wedge \K_i(\psi \to c \approx t) \bigr) \to (\K_i(\varphi \to \psi) \to \Kv_i(\varphi, t)) \).
        Assume $\Kv_i c$, $\K_i(\psi \to c \approx t)$, and $\K_i(\varphi \to \psi)$. From the latter two, we derive $\K_i(\varphi \to c \approx t)$. Then, according to $\text{KVE1}$, $\Kv_i c \wedge \K_i(\varphi \to c \approx t)$ immediately implies $\Kv_i(\varphi, t)$. Applying $\text{Kv}^r$-Elim concludes the proof.
        
        \item By $\text{TAUT}$ and $\text{NECK}$ we get $\K_i(\bot \to 0 \approx t)$. From $\text{Kv}^r0$, we have $\Kv_i 0$. Using $\text{KVE1}$, we derive $\Kv_i(\bot, t)$. 
    \end{enumerate}
\end{proof}

\subsection{\texorpdfstring{Completeness w.r.t. Non-Standard Models}{Completeness w.r.t. Non-Standard Models}}
\label{subsec:comp-nonstandard}

In this subsection, we prove the strong completeness w.r.t. non-standard models. We first consider the classical canonical frame $\mathcal{F}^c = \langle W^c, \{\sim_i^c\}_{i\in I}\rangle$ similar to \cite{Wang2013Knowingthat}:
\begin{itemize}
    \item $W^c= MCS \times \{1,2\}$, where $MCS$ is the set of all maximal consistent sets. 
    \item For any $w,v\in W^c$, $w\sim_i^c v$ iff for any formula $\varphi$, $\K_i\varphi\in w\Rightarrow\varphi\in v$. \footnote{Strictly speaking, $\phi\in w$ means $\phi\in \pi_1(w)$, where $\pi_1$ is the projection function from $W^c$ to $MCS$. We omit $\pi_1$ for simplicity.}
\end{itemize}
Obviously, $\sim_i^c$ is an equivalence relation on $W^c$. The following lemma can be proved by an argument similar to the proof of Proposition 3.10 in \cite{WangFan2014Conditionally}:
    \begin{lemma}
        \label{lem:witness}
        Suppose $\Psi^i_c(w)=\{\psi\mid \Kv_i(\psi,c)\in w\}$. If $\neg \Kv_i(\varphi,c)\in w$, then we can construct $u,v\in MCS$ such that $w\sim_i^c (u,1)$, $w\sim_i^c (v,2)$, $\varphi\in u\cap v$, but $u\cap v\cap \Psi_c^i(w) = \emptyset$.
    \end{lemma}
However, as shown in \cite{Wang2013Knowingthat}, the canonical frame $\mathcal{F}^c$ is not sufficient for the multi-agent completeness proof. Wang and Fan give an insightful but rather involved solution in \cite{WangFan2014Conditionally}. Here we use a different and more direct method: we construct a tree-like frame to facilitate the valuation. The proof is divided into five steps. 
   
\medskip\noindent\textbf{Step 1: Construction of the Tree-like Frame.}
    Given a consistent set of formulas $\Gamma_0$, by Lindenbaum's lemma, we can extend it to a maximal consistent set $\Gamma$. Let $w_0 = (\Gamma, 1)$.
    Construct $\langle W',\{\rightarrow_j\}_{j\in I}\rangle$ level by level. The elements of $W'$ are histories.
       A history is a sequence $\vec{h} = \langle w_0, i_1, w_1, \dots, i_k, w_k\rangle$ where $w_{m-1} \sim_{i_m}^c w_m$ in $\mathcal{F}^c$. Define $\Last(\vec{h}) = w_k$, $\hgt(\vec{h}) = k$, and $\Ag(\langle w_0\rangle)=\star\notin \mathbf{I}$; for $k>0$, $\Ag(\vec{h}) = i_k$.
    In the following, for $d=1,2$,
     $\vec{h}\cdot(j,(w,d))$ denotes the history obtained by appending the agent label $j$ and the canonical world $(w,d)$ to $\vec{h}$. When $\vec{h}'$ is added as a direct $j$-witness for $\vec{h}$, we put $\vec{h}\rightarrow_j\vec{h}'$; the S5 closure is taken in Step 2.
    Let $H_0 = \{\langle w_0\rangle\}$ be the root of the tree. 
    Suppose $H_k$ has been constructed. We construct $H_{k+1}$ as follows.
 For each $\vec{h} \in H_k$ and $j \in I$:
    \begin{enumerate}
        \item If $j = \Ag(\vec{h})$: Skip. The parent of $\vec{h}$ has already constructed the required witnesses for this case, and these will remain accessible to $\vec{h}$ under the S5 closure.
        \item If $j \neq \Ag(\vec{h})$:
        \begin{itemize}
            \item For $\neg \K_j \psi\in \Last(\vec{h})$, by Existence Lemma, we can construct a MCS $w$ such that $\neg\psi\in w$ and $\Last(\vec{h})\sim_j^c w$. Let $\vec{h}' = \vec{h} \cdot (j, (w,1))$ be the new history. Add $\vec{h}'$ to $H_{k+1}$ and let $\vec{h}\rightarrow_j \vec{h}'$.
            \item For $\neg \Kv_j (\phi, t) \in \Last(\vec{h})$: by $\text{Kv}^r\text{S}$, we need only consider the case in which $t$ is a constant symbol $c$. By Lemma \ref{lem:witness}, we can construct $u,v\in MCS$ such that $\Last(\vec{h})\sim_j^c (u,1)$, $\Last(\vec{h})\sim_j^c (v,2)$, $\phi\in u \cap v$ but $u \cap v \cap \Psi^{j}_{c}(\Last(\vec{h})) = \emptyset$. 
            Let $\vec{h}' = \vec{h} \cdot (j, (u,1))$ and $\vec{h}'' = \vec{h} \cdot (j, (v,2))$ be the new histories.
            Add them to $H_{k+1}$ and let $\vec{h}\rightarrow_j \vec{h}'$ and $\vec{h}\rightarrow_j \vec{h}''$.
        \end{itemize}
    \end{enumerate}
    Let $W' = \bigcup_{k\in \mathbb{N}} H_k$. We show that $W'$ is countable. Fix G\"odel codes for formulas and hence for finite sequences of formulas together with labels from $\{1,2\}$. In the construction above, whenever witnesses are required, they are constructed by a fixed deterministic Lindenbaum-style procedure. Each history has at least one construction record, which is a finite sequence: the root is recorded by $1$, and a history of height $k>0$ is recorded as
    \[
        (1,(\theta_1,d_1),\dots,(\theta_k,d_k)),
    \]
    where each $\theta_m$ is the formula whose requirement caused the $m$th successor world to be constructed, and $d_m\in\{1,2\}$ is the corresponding label. A history may have more than one such record, since the same world may witness several formulas. For each history $\vec{h}$, let $m(\vec{h})$ be the least G\"odel code among its construction records. Since a construction record determines the whole history, if $m(\vec{h})=m(\vec{r})$, then $\vec{h}=\vec{r}$. Hence $m$ is injective, and therefore $W'$ is countable.

    \medskip\noindent\textbf{Step 2: $S5$ Closure.}
    Define $\mathcal{F}' = \langle W', \{\sim'_j\}_{j\in I} \rangle$, where $\sim'_j$ is the reflexive, symmetric, transitive closure of $\rightarrow_j$. Due to the construction in Step 1, for each $j\in I$, there are no paths of the form $\vec{h_1} \rightarrow_ j \vec{h_2} \rightarrow_j \vec{h_3}$. Consequently, the structure of $\sim'_j$ equivalence classes is restricted. 
   Each class contains exactly one node $\vec{r}$ with $\Ag(\vec{r}) \neq j$, and all other nodes in the class are immediate $j$-successors of $\vec{r}$.
   Moreover, $\sim'_j$ has the following property.
   \begin{proposition}
    \label{prop:imply}
     If $\vec{h}\sim'_j \vec{r}$, then $\Last(\vec{h})\sim_j^c \Last(\vec{r})$ in the canonical frame $\mathcal{F}^c$, and consequently $\Last(\vec{h})$ and $\Last(\vec{r})$ contain the same $\K_j$-formulas and the same $\Kv_j$-formulas.
   \end{proposition}
  \begin{proof}
     By the construction, if $\vec{h} \rightarrow_j \vec{r}$, then $\Last(\vec{h})\sim_j^c \Last(\vec{r})$.     
     Let $R$ be a binary relation on $W'$ defined by $\vec{x} R \vec{y}$ iff $\Last(\vec{x})\sim_j^c \Last(\vec{y})$. We have $\rightarrow_j \subseteq R$, and $R$ is an equivalence relation on $W'$. Then $\sim'_j\subseteq R$ . The sharing of $\K_j$-formulas follows from axiom 4 and the symmetry of $\sim_j^c$; similarly, the sharing of $\Kv_j$-formulas follows from $\mathrm{Kv}^r4$ and the symmetry of $\sim_j^c$.

  \end{proof}
\noindent\textbf{Step 3: Preparation for Valuation.}
Before defining the valuation function, we first need to determine which constant-history pairs must be placed in the same $\mathbb{Z}$-chain (or in the natural-number part). To this end, we define relations $R_i$ on $(\mathbf{C}\cup \{0\})\times W'$. $(c,\vec{h})R_i(d,\vec{r})\footnote{Note that $c$ and $d$ may also represent $0$.}$ iff one of the following conditions holds:
\begin{itemize}
    \item Type 1: $\vec{h}=\vec{r}$ and there exist $a,b\in\mathbb{N}$ such that $S^a c \approx S^b d \in \Last(\vec{h})$. In this case, we write $(c,\vec{h})\overset{a-b}{\Longrightarrow}(d,\vec{h})$ and also $(d,\vec{h})\overset{b-a}{\Longrightarrow}(c,\vec{h})$.
    \item Type 2: $c$ and $d$ are the same symbol, $\vec{h}\sim_i' \vec{r}$ and there exists a formula $\varphi$ such that $\Kv_i(\varphi, c)\wedge \varphi\in \Last(\vec{h})\cap \Last(\vec{r})$. In this case, $(c,\vec{h})$ and $(d,\vec{r})$ must receive the same value. We write $(c,\vec{h})\overset{\varphi}{\Longrightarrow}(c,\vec{r})$, 
    and also $(c,\vec{r})\overset{0}{\Longrightarrow}(c,\vec{h})$ when the distance is important.
\end{itemize}
Let $\equiv_i$ be the transitive closure of $R_i$. It is an equivalence relation on $(\mathbf{C}\cup \{0\})\times W'$. $(c,\vec{h})\equiv_i (d,\vec{r})$ implies that $\vec{h}\sim_i' \vec{r}$ and thus $\Last(\vec{h})\sim_i^c \Last(\vec{r})$.
By arithmetic axioms, $(c,\vec{h})\overset{a}{\Longrightarrow}(d,\vec{h})\overset{b}{\Longrightarrow}(e,\vec{h})$ implies $(c,\vec{h})\overset{a+b}{\Longrightarrow}(e,\vec{h})$. 
If $(c,\vec{h})\overset{\varphi}{\Longrightarrow}(c,\vec{r})\overset{\psi}{\Longrightarrow}(c,\vec{u})$, then $\Khat_i(\varphi\wedge\psi)\wedge \Kv_i(\varphi,c)\wedge \Kv_i(\psi,c)\in \Last(\vec{h})\cap \Last(\vec{r})\cap \Last(\vec{u})$; by $\mathrm{Kv}^r\vee$, this implies $\Kv_i(\varphi\vee\psi,c)\wedge(\varphi\vee\psi)\in \Last(\vec{h})\cap \Last(\vec{u})$. 
Thus $(c,\vec{h})\overset{\varphi\vee\psi}{\Longrightarrow}(c,\vec{u})$. Now it is time to show the power of $\mathrm{GKv}^r\vee$ and $\mathrm{GKVE}$.
\begin{lemma}
    \label{lemm: loop}
    For any $(c,\vec{h}),(d,\vec{h})\in (\mathbf{C}\cup \{0\})\times W'$, if $(c,\vec{h})\equiv_i (d,\vec{h})$, i.e. there exist $(c_0,\vec{h}_0),\dots,(c_n,\vec{h}_n)$ such that $(c,\vec{h})=(c_0,\vec{h}_0)\overset{a_1}{\Longrightarrow}(c_1,\vec{h}_1)\overset{a_2}{\Longrightarrow}\cdots \overset{a_n}{\Longrightarrow}(c_n,\vec{h}_n)=(d,\vec{h})$, where $a_1,\dots,a_n\in \mathbb{Z}$ and $\sum_k a_k = a$, then $S^a c \approx d \in \Last(\vec{h})$ if $a\ge 0$, and $c \approx S^{-a} d \in \Last(\vec{h})$ if $a<0$.
\end{lemma}
\begin{proof}
    We first normalize the witnessing chain. Each step is either:
    (i) a type 1 step in a single world node, written $(c,\vec{r})\overset{a}{\Longrightarrow}(d,\vec{r})$, or
    (ii) a type 2 step between $i$-equivalent world nodes, written $(c,\vec{r})\overset{\varphi}{\Longrightarrow}(c,\vec{u})$.
    Two consecutive type 1 steps can be merged into one by arithmetic axioms, and two consecutive type 2 steps can be merged into one by $\mathrm{Kv}^r\vee$. 
    If there are only type 1 steps, then the conclusion follows immediately from arithmetic axioms. 
    Now we present the case where the chain starts and ends with a type 2 step; if it starts (or ends) with a type 1 step, simply remove that initial (or terminal) type 1 step, apply the argument below, and then reattach it using arithmetic axioms. Thus we may assume an alternating chain of the form
    \[
    (c_0,\vec{h}_0)\overset{\phi_0}{\Longrightarrow}(c_0,\vec{h}_1)\overset{a_0}{\Longrightarrow}(c_1,\vec{h}_1)\overset{\phi_1}{\Longrightarrow}(c_1,\vec{h}_2)\overset{a_1}{\Longrightarrow}\cdots\overset{a_{n-1}}{\Longrightarrow}(c_n,\vec{h}_n)\overset{\phi_n}{\Longrightarrow}(c_n,\vec{h}_0),
    \]
    where $a_0,\dots,a_{n-1}\in \mathbb{Z}$, $\phi_0,\dots,\phi_n$ are formulas, and $\sum_{i=0}^{n-1} a_i=a$. Note that $\vec{h}_0=\vec{h}$ and all $\vec{h}_s$ are $\sim'_i$-equivalent. Then all the $\Last(\vec{h}_s)$ are $\sim_i^c$-equivalent, and thus share the same 
    $\K_i$-formulas and $\Kv_i$-formulas. 

    \noindent Let $s_k:=\sum_{x=k}^{n-1} a_x$ for $k=0,\dots,n$, and let $K:=\min\{s_k\mid 0\le k\le n\}$. Define
    \( t_k:= S^{s_k-K} c_k. \)
    Then all exponents are non-negative. Since $s_n=0$, we have $t_n=S^{-K}c_n$, and $t_0=S^{a-K}c_0$.

    For each $k<n$, the type 1 step $(c_k,\vec{h}_{k+1})\overset{a_k}{\Longrightarrow}(c_{k+1},\vec{h}_{k+1})$ yields $t_k\approx t_{k+1}\in \Last(\vec{h}_{k+1})$ by applying successor on both sides the appropriate number of times. Also, the type 2 step $(c_k,\vec{h}_k)\overset{\phi_k}{\Longrightarrow}(c_k,\vec{h}_{k+1})$ gives $\phi_k\in \Last(\vec{h}_k)\cap \Last(\vec{h}_{k+1})$ and $\Kv_i(\phi_k,c_k)\in \Last(\vec{h}_k)\cap \Last(\vec{h}_{k+1})$.
    By $\mathrm{Kv}^r\mathrm{S}$, $\Kv_i(\phi_k,t_k)\in \Last(\vec{h}_k)\cap \Last(\vec{h}_{k+1})$.
     Additionally, $\phi_n\wedge \Kv_i(\phi_n,t_n)\in \Last(\vec{h}_n)\cap \Last(\vec{h}_0)$. 
    Since all the $\Last(\vec{h}_s)$ share the same $\Kv_i$-formulas, $\Kv_i(\phi_k,t_k)\in \Last(\vec{h}_0)$ for each $k=0,\dots,n$.
Moreover, each $\Last(\vec{h}_{k+1})$ contains $\phi_k\wedge\phi_{k+1}\wedge t_k\approx t_{k+1}$ for $k<n$, so $\widehat{K}_i(\phi_k\wedge\phi_{k+1}\wedge t_k\approx t_{k+1})\in \Last(\vec{h}_0)$ for all $k<n$. Therefore, the premises of $\mathrm{GKVE}$ hold in $\Last(\vec{h}_0)$, and we obtain
    \( \K_i(\phi_0\wedge\phi_n\to t_0\approx t_n)\in \Last(\vec{h}_0). \)
    Since $\phi_0\wedge\phi_n\in \Last(\vec{h}_0)$, we get $t_0\approx t_n\in \Last(\vec{h}_0)$.
Finally, $t_0\approx t_n$ is $S^{a-K}c\approx S^{-K}d$. Cancelling $S^{-K}$ on both sides using arithmetic axioms repeatedly yields $S^a c\approx d$ when $a\ge 0$, and $c\approx S^{-a} d$ when $a<0$. Hence the lemma holds.
\end{proof}

\begin{corollary}
    \label{coro:loop}
    For any loop of the form $(c,\vec{h})=(c_0,\vec{h}_0)\overset{a_1}{\Longrightarrow}(c_1,\vec{h}_1)\overset{a_2}{\Longrightarrow}\cdots \overset{a_n}{\Longrightarrow}(c_n,\vec{h}_n)=(c,\vec{h})$, where $a_1,\dots,a_n\in \mathbb{Z}$, we have $\sum_k a_k = 0$. 
    Consequently, if there are two paths from $(c,\vec{h})$ to $(d,\vec{r})$, one with total weight $a$ and the other with total weight $b$, then $a=b$.
\end{corollary}
\begin{proof}
    By the above lemma, we have $S^{\sum_k a_k} c \approx c \in \Last(\vec{h})$ if $\sum_k a_k \ge 0$ and $c \approx S^{-\sum_k a_k} c \in \Last(\vec{h})$ if $\sum_k a_k < 0$. 
    But we also have $\neg S^n c\approx c\in \Last(\vec{h})$ for $n>0$. Since $\Last(\vec{h})$ is consistent, we have $\sum_k a_k = 0$.
\end{proof}
\begin{lemma}
    \label{lemm:loop2}
    For any $(c,\vec{h}),(c,\vec{r})\in (\mathbf{C}\cup \{0\})\times W'$, if $(c,\vec{h})\equiv_i (c,\vec{r})$, i.e. there exist $(c_0,\vec{h}_0),\dots,(c_n,\vec{h}_n)$ such that $(c,\vec{h})=(c_0,\vec{h}_0)\overset{a_1}{\Longrightarrow}(c_1,\vec{h}_1)\overset{a_2}{\Longrightarrow}\cdots \overset{a_n}{\Longrightarrow}(c_n,\vec{h}_n)=(c,\vec{r})$, where $a_1,\dots,a_n\in \mathbb{Z}$ and $\sum_k a_k = 0$, then $(c,\vec{h})\overset{0}{\Longrightarrow}(c,\vec{r})$, i.e. if $\vec{h}\neq \vec{r}$, then there exists a formula $\varphi$ such that $\Kv_i(\varphi, c)\wedge \varphi\in \Last(\vec{h})\cap \Last(\vec{r})$.
\end{lemma}
\begin{proof}
    We first normalize the path as before. We present the case where the chain starts and ends with a type 1 step.
    If the chain starts or ends with a type 2 step, insert the trivial type 1 step of weight $0$ (using $c\approx c$) so that the chain starts and ends with type 1 steps. Thus we may assume the chain has the alternating form
    \[
    (c_0,\vec{h}_0)\overset{a_0}{\Longrightarrow}(c_1,\vec{h}_0)\overset{\phi_1}{\Longrightarrow}(c_1,\vec{h}_1)\overset{a_1}{\Longrightarrow}(c_2,\vec{h}_1)\overset{\phi_2}{\Longrightarrow}\cdots\overset{\phi_n}{\Longrightarrow}(c_n,\vec{h}_n)\overset{a_n}{\Longrightarrow}(c_0,\vec{h}_n),
    \]
    where $a_0,\dots,a_n\in \mathbb{Z}$, $\phi_1,\dots,\phi_n$ are formulas, and $\sum_{k=0}^n a_k=0$. Note that all $\vec{h}_k$ are $\sim'_i$-equivalent.

    Let $s_k:=\sum_{x=k}^n a_x$ for $k=0,\dots,n$, and let $K:=\min\{s_k\mid 0\le k\le n\}$. Define
    \( t_k:= S^{s_k-K} c_k. \)
    Then all exponents are non-negative. Since $s_0=0$, we have $t_0=S^{-K}c_0$. $(c_n,\vec{h}_n)\overset{a_n}{\Longrightarrow}(c_0,\vec{h}_n)$ implies $t_n\approx t_0\in \Last(\vec{h}_n)$.
    For each $k<n$, the type 1 step $(c_k,\vec{h}_k)\overset{a_k}{\Longrightarrow}(c_{k+1},\vec{h}_k)$ implies $t_k\approx t_{k+1}\in \Last(\vec{h}_k)$. Also, for each $k=1,\dots,n$, since $(c_k,\vec{h}_{k-1})\overset{\phi_k}{\Longrightarrow}(c_k,\vec{h}_k)$, we have $\phi_k\in \Last(\vec{h}_{k-1})\cap \Last(\vec{h}_k)$ and $\Kv_i(\phi_k,c_k)\in \Last(\vec{h}_{k-1})\cap \Last(\vec{h}_k)$. By $\mathrm{Kv}^r\mathrm{S}$, $\Kv_i(\phi_k,t_k)\in \Last(\vec{h}_{k-1})\cap \Last(\vec{h}_k)$. Since all the $\Last(\vec{h}_s)$ share the same $\Kv_i$-formulas, $\Kv_i(\phi_k,t_k)\in \Last(\vec{h}_0)\cap \Last(\vec{h}_n)$ for every $k$.

    Moreover, $\Last(\vec{h}_k)$ contains $\phi_k\wedge\phi_{k+1}\wedge t_k\approx t_{k+1}$ for $1\leq k<n$, so since $\vec{h}_0\sim'_i \vec{h}_k$ and $\vec{h}_n\sim'_i \vec{h}_k$, we have
    \( \widehat{K}_i(\phi_k\wedge\phi_{k+1}\wedge t_k\approx t_{k+1})\in \Last(\vec{h}_0)\cap \Last(\vec{h}_n). \)
    Therefore, the premises of $\mathrm{GKv}^r\vee$ are satisfied in both $\Last(\vec{h}_0)$ and $\Last(\vec{h}_n)$ for the sequence $(\phi_1,t_1),\dots,(\phi_n,t_n)$. Hence
    \( \Kv_i\Bigl( \bigvee_{k=1}^n (\phi_k \wedge t_0 \approx t_k), t_0 \Bigr) \in \Last(\vec{h}_0)\cap \Last(\vec{h}_n). \)
    Let $\psi:= \bigvee_{k=1}^n (\phi_k \wedge t_0 \approx t_k)$. We have $\phi_1\wedge t_0\approx t_1\in \Last(\vec{h}_0)$ and $\phi_n\wedge t_0\approx t_n\in \Last(\vec{h}_n)$, so $\psi\in \Last(\vec{h}_0)\cap \Last(\vec{h}_n)$. By $\mathrm{Kv}^r\mathrm{S}$, $\Kv_i(\psi,c_0)\in \Last(\vec{h}_0)\cap \Last(\vec{h}_n)$. Thus $(c_0,\vec{h}_0)\overset{\psi}{\Longrightarrow}(c_0,\vec{h}_n)$.
\end{proof}
\noindent\textbf{Step 4: Valuation Function.}
Our goal now is to define the valuation function $V_\mathbf{C}^\star$ on $W'$. Note that $W'$ is countable, and for each agent $i$, every $\sim_i'$-equivalence class in $W'$ forms a cluster consisting of exactly one parent node and all its immediate $i$-successors. This allows us to define the valuation level by level.
Fix an explicit partition of the $\mathbb{Z}$-chains into pairwise disjoint countable sets $\{E_n\}_{n\in\mathbb{N}}$ as follows. Enumerate all $\mathbb{Z}$-chains as $\{Z_k\mid k\in\mathbb{N}\}$, and fix a bijection $\langle\cdot,\cdot\rangle:\mathbb{N}\times\mathbb{N}\to\mathbb{N}$. Then define
\[ E_n:=\{Z_{\langle n,m\rangle}\mid m\in\mathbb{N}\}. \]
Each $E_n$ is countably infinite and the $E_n$ are pairwise disjoint, and $\bigcup_n E_n$ is the set of all $\mathbb{Z}$-chains. Intuitively, $E_n$ is the pool of fresh $\mathbb{Z}$-chains reserved for equivalence classes that first appear when we define the valuation on histories of height $n$.

Now define $V_\mathbf{C}^\star$ on $W'$ as follows: 
For $\langle w_0\rangle \in H_0$, intuitively, we only respect the atomic formulas in $w_0$. Technically, we can consider $\equiv_i$ classes in $(\mathbf{C}\cup\{0\})\times \{\langle w_0\rangle\}$ for arbitrary $i$. 
 For each equivalence class, we can find a representative. If $(0,\langle w_0\rangle)$ is in this class, $(0,\langle w_0\rangle)$ is the representative; otherwise, 
    $(c_n,\langle w_0\rangle)$ is the representative, where $c_n$ is the constant symbol with the smallest index in this class. The equivalence class $[0,\langle w_0\rangle]$ is placed in the natural-number part. $[(c_n,\langle w_0\rangle)]$ is placed in $Z_{\langle 0,2^{n+1}\rangle}$ (the representative receives value 0 in this $\mathbb{Z}$-chain). For any other $(d,\langle w_0\rangle)$ in this class, there are paths from the representative to it. However, by Corollary \ref{coro:loop}, all these paths have the same sum of weights. Then we define $V_\mathbf{C}^\star(d,\langle w_0\rangle)$ as the value of the representative plus the sum of weights on the path within the same $\mathbb{Z}$-chain. 
    
First note that for any two pairs $(c,\langle w_0\rangle)$ and $(d,\langle w_0\rangle)$,
if there is a $z$-weight path from one to the other, then $V^\star_\mathbf{C}(c,\langle w_0\rangle)+z=V^\star_\mathbf{C}(d,\langle w_0\rangle)$. This is because 
if there is an $x$-weight path from the representative to $(c,\langle w_0\rangle)$ and a $y$-weight path from the representative to $(d,\langle w_0\rangle)$, then there is a path from $(c,\langle w_0\rangle)$ to $(d,\langle w_0\rangle)$ with weight $y-x$. By Corollary \ref{coro:loop}, we have $y-x=z$, and thus $V^\star_\mathbf{C}(c,\langle w_0\rangle)+z=V^\star_\mathbf{C}(d,\langle w_0\rangle)$.

Now we prove that $S^a c\approx S^b d\in w_0$ iff $V^\star_\mathbf{C}(c,\langle w_0\rangle)+a=V^\star_\mathbf{C}(d,\langle w_0\rangle)+b$ (where $c,d$ may represent $0$). 
First if $S^a c\approx S^b d\in w_0$, then $(c,\langle w_0\rangle)\overset{a-b}{\Longrightarrow}(d,\langle w_0\rangle)$, by the above observation, $V^\star_\mathbf{C}(c,\langle w_0\rangle)+a= V^\star_\mathbf{C}(d,\langle w_0\rangle)+b$.
If $\neg S^a c \approx S^b d\in w_0$, then if there is a $x$-weight path from $(c,\langle w_0\rangle)$ to $(d,\langle w_0\rangle)$, then by the observation above, $V^\star_\mathbf{C}(c,\langle w_0\rangle)+x=V^\star_\mathbf{C}(d,\langle w_0\rangle)$. But by Lemma \ref{lemm: loop}, $x\neq a-b$. Thus $V^\star_\mathbf{C}(c,\langle w_0\rangle)+a\neq V^\star_\mathbf{C}(d,\langle w_0\rangle)+b$. If there is no path between them, they will be placed in different $\mathbb{Z}$-chains, and thus $V^\star_\mathbf{C}(c,\langle w_0\rangle)+a\neq V^\star_\mathbf{C}(d,\langle w_0\rangle)+b$ as well.

     Assume we have defined the values of constants for all histories with height $< n$, and we have used only chains from $E_0,\dots,E_{n-1}$ so far. Suppose also that we have verified that for any history $\vec{h}$ with height $< n$, $S^a c\approx S^b d\in \Last(\vec{h})$ iff $V^\star_\mathbf{C}(c,\vec{h})+a=V^\star_\mathbf{C}(d,\vec{h})+b$ (where $c,d$ may represent $0$).
    
      Now consider any $\sim'_j$ equivalence class in $H_{n-1}\cup H_n$. It contains one parent $\vec{p}\in H_{n-1}$ where the values of constants have already been defined, and countably many children in $H_n$ which are immediate $j$-successors of $\vec{p}$. Then we consider the equivalence classes by $\equiv_j$ over this $\sim'_j$ equivalence class. 
      The valuation of $\vec{p}$ is compatible: by the assumption, if $(c,\vec{p})$ and $(d,\vec{p})$ (where $c,d$ may represent $0$) have been placed in the same component (both in $\mathbb{N}$ or both in the same $\mathbb{Z}$-chain)
      with relative distance $z$, then they are connected in $\equiv_j$ by a path of weight $z$. If they are placed in different components, they must be $\equiv_j$ inequivalent. Therefore, for any $\equiv_j$ equivalence class containing some pair from $\vec{p}$, if it contains $(0,\vec{p})$, we take $(0,\vec{p})$ as the representative, place the whole class in the natural-number part, and give the representative value $0$. By $\mathrm{Kv}^r0$ with $\top$ as the witnessing formula, every $(0,\vec{u})$ in this $\sim'_j$-class belongs to this same $\equiv_j$-class. Otherwise, we choose the representative to be $(c_k,\vec{p})$ where $k$ is the smallest index such that $(c_k,\vec{p})$ is in this class, and we place the elements in this class in the same component as $(c_k,\vec{p})$, keeping the already defined value of $(c_k,\vec{p})$.
      
      For each $\equiv_j$ equivalence class that does not contain any pair of the form $(c,\vec{p})$,
    by the previous paragraph, such a class contains no pair of the form $(0,\vec{u})$. We can also find a representative because histories were encoded as natural numbers above. Specifically, choose the pair $(c_{k},\vec{v})$ in the class with lexicographically least $(k,m(\vec{v}))$, where $m(\vec{v})$ is the code of $\vec{v}$. 
    $[(c_{k},\vec{v})]$, the $\equiv_j$-class of this representative, is placed in $Z_{\langle n,2^{k+1}\cdot 3^{m(\vec{v})}\rangle}$. Also, $(c_{k},\vec{v})$ receives value 0 in this $\mathbb{Z}$-chain.

    In both kinds of classes, we have fixed where the representative is placed in $\mathbb{N}^\star$, either in the natural-number part or in a specified $\mathbb{Z}$-chain, and we have also fixed its value there.
    As before, for any other $(d,\vec{r})$ in the class, we define its value as the value of the representative plus the sum of weights on any path from the representative to it. 
    Then we can similarly verify that for any two pairs whose values have been defined,
    if there is a $z$-weight path from one to the other, then their values differ by $z$.
    Moreover, for any history $\vec{h}$ with height $\leq n$, $S^a c\approx S^b d\in \Last(\vec{h})$ iff $V^\star_\mathbf{C}(c,\vec{h})+a=V^\star_\mathbf{C}(d,\vec{h})+b$ (where $c,d$ may represent $0$).
Finally, throughout the construction, each pair $(c,\vec{h})$ receives exactly one value. This is because when $\vec{h}$ is a child in the equivalence class, we define its value; when $\vec{h}$ is a parent, we respect the value already defined for it.

\medskip\noindent\textbf{Step 5: Truth Lemma.} We can prove the truth lemma by induction on formulas.
\begin{lemma}[Truth Lemma]
    Let $\mathcal{M}:=\langle W',\mathbb{N}^\star, \mathbf{S}^\star,\{\sim_i'\}_{i\in I},V_\mathbf{C}^\star\rangle$ be the model constructed above.
    For any formula $\varphi$ and any history $\vec{h}\in W'$, $\mathcal{M},\vec{h}\models \varphi$ iff $\varphi\in \Last(\vec{h})$.   
\end{lemma}\begin{proof}
     By induction on $\varphi$, we prove that for any $\vec{h}\in W'$, $\mathcal{M},\vec{h}\models \varphi$ iff $\varphi\in \Last(\vec{h})$.
    The atomic case was proved when we defined $V_\mathbf{C}^\star$.
    The boolean cases are trivial.
    
    For the case $\varphi := \K_i \psi$: if $\K_i \psi \in \Last(\vec{h})$, then for any $\vec{u} \in W'$ with $\vec{h} \sim'_i \vec{u}$, we have $\Last(\vec{h}) \sim^c_i \Last(\vec{u})$ by Proposition \ref{prop:imply}. Thus $\psi \in \Last(\vec{u})$. The induction hypothesis gives $\mathcal{M}, \vec{u} \models \psi$, so $\mathcal{M}, \vec{h} \models \K_i \psi$. On the other hand, if $\neg \K_i \psi \in \Last(\vec{h})$, we consider two cases. If $i \neq \Ag(\vec{h})$, there exists a direct witness $\vec{u} \in W'$ such that $\vec{h} \rightarrow_i \vec{u}$ and $\neg \psi \in \Last(\vec{u})$.
    If $i = \Ag(\vec{h})$, the witness construction for $i$ was skipped for $\vec{h}$. However, consider its $i$-parent $\vec{p}$, since $\Last(\vec{h}) \sim_i^c \Last(\vec{p})$, $\neg \K_i \psi \in \Last(\vec{p})$. By the construction of $W'$, we
    have selected a witness $\vec{u}$ such that $\vec{p} \rightarrow_i \vec{u}$ and $\neg \psi \in \Last(\vec{u})$. In both cases, $\vec{h} \sim'_i \vec{u}$ and $\neg \psi \in \Last(\vec{u})$. By the induction hypothesis, $\mathcal{M}, \vec{u} \not\models \psi$, establishing $\mathcal{M}, \vec{h} \not\models \K_i \psi$.
    
    For the case $\varphi := \Kv_i(\psi, t)$, by applying axiom $\text{Kv}^r \text{S}$, it suffices to consider the case in which $t$ is a constant $c$. If $\Kv_i(\psi, c) \in \Last(\vec{h})$, let $\vec{u}, \vec{v} \in W'$ be such that $\vec{h} \sim'_i \vec{u}$, $\vec{h} \sim'_i \vec{v}$ and both $\mathcal{M}, \vec{u} \models \psi$ and $\mathcal{M}, \vec{v} \models \psi$. By the induction hypothesis, $\psi \in \Last(\vec{u}) \cap \Last(\vec{v})$. Since $\Last(\vec{h}) \sim_i^c \Last(\vec{u})\sim_i^c \Last(\vec{v})$, 
    $\Kv_i(\psi, c) \in \Last(\vec{u}) \cap \Last(\vec{v})$. We also have $\vec{u} \sim'_i \vec{v}$. The definition of $R_i$ then guarantees $(c, \vec{u}) \overset{\psi}{\Longrightarrow} (c, \vec{v})$. Thus, $(c, \vec{u}) \equiv_i (c, \vec{v})$ with weight 0. By the construction of $V_\mathbf{C}^\star$, $V_\mathbf{C}^\star(c, \vec{u}) = V_\mathbf{C}^\star(c, \vec{v})$, yielding $\mathcal{M}, \vec{h} \models \Kv_i(\psi, c)$. 
    Conversely, if $\neg \Kv_i(\psi, c) \in \Last(\vec{h})$, we again consider two cases. If $i \neq \Ag(\vec{h})$, the construction directly provides witnesses $\vec{u}, \vec{v} \in W'$ such that $\vec{h} \rightarrow_i \vec{u}$ and $\vec{h} \rightarrow_i \vec{v}$, with $\psi \in \Last(\vec{u}) \cap \Last(\vec{v})$ and $\Last(\vec{u}) \cap \Last(\vec{v}) \cap \Psi^i_c(\Last(\vec{h})) = \emptyset$. If $i = \Ag(\vec{h})$, the witnesses were selected by its $i$-parent $\vec{p}$; hence there are $\vec{u}$ and $\vec{v}$ such that $\vec{p} \rightarrow_i \vec{u}$ and $\vec{p} \rightarrow_i \vec{v}$, with $\psi \in \Last(\vec{u}) \cap \Last(\vec{v})$ and $\Last(\vec{u}) \cap \Last(\vec{v}) \cap \Psi^i_c(\Last(\vec{h})) =\Last(\vec{u}) \cap \Last(\vec{v}) \cap \Psi^i_c(\Last(\vec{p})) = \emptyset$. 
In either case, we have $\vec{h} \sim'_i \vec{u}$, $\vec{h} \sim'_i \vec{v}$, $\psi \in \Last(\vec{u}) \cap \Last(\vec{v})$ and $\Last(\vec{u}) \cap \Last(\vec{v}) \cap \Psi^i_c(\Last(\vec{h})) = \emptyset$.
By the induction hypothesis, $\mathcal{M}, \vec{u} \models \psi$ and $\mathcal{M}, \vec{v} \models \psi$. 
If $(c, \vec{u}) \equiv_i (c, \vec{v})$ with weight 0, then by Lemma \ref{lemm:loop2}, there exists a formula $\chi$ such that $\Kv_i(\chi, c) \wedge \chi \in \Last(\vec{u}) \cap \Last(\vec{v})$. Then $\chi\in \Psi^i_c(\Last(\vec{h}))$, contradicting the choice of $\vec{u}$ and $\vec{v}$. If it is not the case that $(c, \vec{u}) \equiv_i (c, \vec{v})$ with weight 0, then by the construction of $V_\mathbf{C}^\star$, $V_\mathbf{C}^\star(c, \vec{u}) \neq V_\mathbf{C}^\star(c, \vec{v})$, yielding $\mathcal{M}, \vec{h} \not\models \Kv_i(\psi, c)$.
\end{proof}
This yields the desired completeness result.
\medskip
\begin{theorem}
    $\sys$ is strongly complete w.r.t. non-standard models.
\end{theorem}
\begin{proof}
    Given a consistent formula set $\Gamma_0$, $\mathcal{M},\langle (\Gamma,1)\rangle \models \Gamma_0$.
\end{proof}

\section{From Non-standard Models to Standard Models}
\label{sec:standard-models}
In this section, we bridge the gap between non-standard models and standard models. This allows us to establish the weak completeness of $\sys$ with respect to standard models.
In the process, we also obtain the finite model property w.r.t. both standard and non-standard models,
which will be used to establish
 the decidability of $\logic$.
We first show that satisfiability in a finite non-standard model implies satisfiability in a finite standard model.
\begin{lemma}
    \label{lemm:standard}
    For any formula $\phi$ satisfiable in a finite non-standard model $\mathcal{N}$, 
 it is also satisfiable in a finite standard model.
\end{lemma}

\begin{proof}
	Let $\mathcal{N} = \langle W,\mathbb{N}^\star, \mathbf{S}^\star, \{\sim_i\}_{i\in \mathbf{I}}, V^\star_\mathbf{C}\rangle$ be a non-standard model such that $\mathcal{N}, w_0 \models \phi$, where $W$ is finite. 
	Let
    \[ U=\{V_\mathbf{C}^{\star}(c,w)\mid c \text{ occurs in }\phi,\ w\in W\}\cup\{0\}. \]
	Intuitively, $U$ is the set of all the numbers that are relevant to the truth of $\phi$ in $\mathcal{N}$. 
    Since $W$ is finite and $\phi$ contains only finitely many constants, $U$ is finite and thus $U\cap Z_n$ and $U\cap\mathbb{N}$ are finite for each $n$.

	Let $B= \max(\{a+b\mid \exists  x,y\in \mathbf{C}\cup\{0\},\ S^a x \approx S^b y  \text{ is a subformula of }\phi\}\cup \{0\})$.
	We build an injective map $f:U\to \mathbb{N}$ such that:
	\begin{enumerate}
		\item For $x,y\in U$ in the same component (both in $\mathbb{N}$ or both in the same $Z_n$),
		$f(x)-f(y)=x-y$.
		\item For $x,y\in U$ in different components, $|f(x)-f(y)|>B$.
		\item $f(x)=x$ for all $x\in U\cap\mathbb{N}$.
	\end{enumerate}
	Such an $f$ exists because $U\cap \mathbb{N}$ and each $U\cap Z_n$ is finite, so we can preserve all local distances inside a component and separate different components by gaps greater than $B$.
Define a standard model $\mathcal{M} = \langle W,\mathbb{N}, \mathbf{S}, \{\sim_i\}_{i\in \mathbf{I}}, V_\mathbf{C}\rangle$ by
	defining \( V_\mathbf{C}(c,u)=f\bigl(V_\mathbf{C}^{\star}(c,u)\bigr) \)
	for each constant $c$ occurring in $\phi$ and $u\in W$. For other constants, assign arbitrary values.
    We show by induction on subformulas $\psi$ of $\phi$ that for all $u\in W$,
	\[ \mathcal{N},u\models \psi \iff \mathcal{M},u\models \psi. \]
    \begin{itemize}
		\item Atomic case: $\psi := t_1 \approx t_2$. Suppose $t_1$ and $t_2$ are $S^a c_1$ and $S^b c_2$ respectively (the case of $S^k0$ can be proved similarly). If $\mathcal{N}, u \models t_1 \approx t_2$, then $V_\mathrm{T}(t_1,u) = V_\mathrm{T}(t_2,u)$, which means $V_\mathbf{C}^\star(c_1,u) + a = V_\mathbf{C}^\star(c_2,u) + b$. By the definition of $f$, we have $f(V_\mathbf{C}^\star(c_1,u)) + a = f(V_\mathbf{C}^\star(c_2,u)) + b$, which means $\mathcal{M}, u \models t_1 \approx t_2$. 
		Suppose $\mathcal{N}, u \models \neg S^a c_1\approx S^b c_2$, then if $V_\mathbf{C}^\star(c_1,u)$ and $V_\mathbf{C}^\star(c_2,u)$ are in the same component but $V_\mathbf{C}^\star(c_1,u) + a \neq V_\mathbf{C}^\star(c_2,u) + b$, then by property 1, $f(V_\mathbf{C}^\star(c_1,u)) + a \neq f(V_\mathbf{C}^\star(c_2,u)) + b$ and $\mathcal{M}, u \models \neg t_1 \approx t_2$. If $V_\mathbf{C}^\star(c_1,u)$ and $V_\mathbf{C}^\star(c_2,u)$ are in different components, then by property 2, $|f(V_\mathbf{C}^\star(c_1,u)) - f(V_\mathbf{C}^\star(c_2,u))| > B \geq |a-b|$, so $f(V_\mathbf{C}^\star(c_1,u)) + a \neq f(V_\mathbf{C}^\star(c_2,u)) + b$ and $\mathcal{M}, u \models \neg t_1 \approx t_2$.
        \item Boolean connectives and $\psi:=\top$: Trivial.
        \item $\psi := \K_i \chi$. 
        The IH is: $\forall v \in W, \mathcal{N}, v \models \chi \iff \mathcal{M}, v \models \chi$.
        \begin{align*}
            \mathcal{N}, u \models \K_i \chi &\iff \forall v \in W (u \sim_i v \Rightarrow \mathcal{N}, v \models \chi) \\
            &\iff \forall v \in W (u \sim_i v \Rightarrow \mathcal{M}, v \models \chi) \quad (\text{by IH}) \\
            &\iff \mathcal{M}, u \models \K_i \chi
        \end{align*}
        \item $\psi := \Kv_i(\chi,t)$. By the semantics, $\Kv_i(\chi,S^k c) \leftrightarrow \Kv_i(\chi,c)$ and $\Kv_i(\chi, S^k 0) \leftrightarrow \top$ are valid. So we only need to consider the case in which $t$ has the form $c$. Note that the IH is: $\forall v \in W, \mathcal{N}, v \models \chi \iff \mathcal{M}, v \models \chi$.
        \begin{align*}
            \mathcal{N}, u \models \Kv_i(\chi,c) 
            \iff& \forall v_1, v_2 \sim_i u, \text{ if } \mathcal{N}, v_1 \models \chi \text{ and } \mathcal{N}, v_2 \models \chi,\\
			&\text{then } V^\star_\mathbf{C}(c, v_1) = V^\star_\mathbf{C}(c, v_2) \\
            \iff& \forall v_1, v_2 \sim_i u, \text{if } \mathcal{M}, v_1 \models \chi \text{ and } \mathcal{M}, v_2 \models \chi,\\
			&\text{then }
			f(V^\star_\mathbf{C}(c, v_1)) = f(V^\star_\mathbf{C}(c, v_2)) \quad (\text{IH and } f \text{ is injective}) \\
            \iff& \forall v_1, v_2 \sim_i u, \text{if } \mathcal{M}, v_1 \models \chi \text{ and } \mathcal{M}, v_2 \models \chi,\\
			&\text{then }
			V_\mathbf{C}(c, v_1) = V_\mathbf{C}(c, v_2) \\
            \iff& \mathcal{M}, u \models \Kv_i(\chi,c)
        \end{align*}
    \end{itemize}
	Hence $\mathcal{M},w_0\models \phi$, so $\phi$ is satisfiable in a finite standard model.
\end{proof}

To apply the above lemma, we prove the finite model property for non-standard models via a selection argument from the canonical tree-like model.
\begin{theorem}
    \label{thm:fmp}
    $\logic$ has finite model property w.r.t. non-standard models. Moreover, if a formula $\phi$ is satisfiable in a non-standard model, it is satisfiable in a finite non-standard model whose size is bounded by $(2|\Sigma|)^{\dep(\phi)+1}$, where $\Sigma$ is the set of subformulas of $\phi$ closed under single negation, and $\dep(\phi)$ is the modal depth of $\phi$.
\end{theorem}
\begin{proof}
    For any formula $\phi$ satisfiable in a non-standard model, it is satisfiable in the tree-like model constructed in the completeness proof in Subsection \ref{subsec:comp-nonstandard}. 
   This tree-like model is $\mathcal{M}= \langle W', \mathbb{N}^\star, \mathbf{S}^\star, \{\sim_i'\}_{i\in I}, V_\mathbf{C}^\star\rangle$. 
   Suppose the root is $\langle w_0 \rangle$, where $\phi \in w_0$.
   Our goal is to select a finite submodel from $\mathcal{M}$ that still satisfies $\phi$. 
    Let $\Sigma$ be the set of subformulas of $\phi$ closed under single negation.
    We define the \textbf{modal depth} $\dep(\varphi)$ as:
    \begin{itemize}
        \item $\dep(\alpha) = 0$ for atomic formulas and Boolean combinations of atomics.
        \item $\dep(\neg \varphi) = \dep(\varphi)$, $\dep(\varphi \wedge \psi) = \max(\dep(\varphi), \dep(\psi))$.
        \item $\dep(\K_i \varphi) = \dep(\varphi) + 1$.
        \item $\dep(\Kv_i(\varphi,t)) = \dep(\varphi)+1$.
    \end{itemize}
    Let $D = \dep(\phi)$.
    
    \medskip\noindent\textbf{Selection.}
    Construct $W_s \subset W'$ level by level. Recall that $W'=\bigcup_{k=0}^\infty H_k$ where $H_k$ is the set of nodes with height $k$.
    Let $H'_0 = H_0 =\{\langle w_0\rangle\}$. Assume $H'_k$ is constructed for $k < D$.
    We construct $H'_{k+1}$ as follows. For each $\vec{h} \in H'_k$ and $j \in I$:
    \begin{enumerate}
        \item If $j = \Ag(\vec{h})$: Skip. 
        \item If $j \neq \Ag(\vec{h})$:
        \begin{itemize}
            \item For $\neg \K_j \psi \in \Sigma$: If $\mathcal{M}, \vec{h} \models \neg\K_j \psi$ and $\dep(\neg \K_j \psi) \le D - k$, select one witness $\vec{h}'\in H_{k+1}$ with $\mathcal{M}, \vec{h}' \nvDash \psi$ and add to $H'_{k+1}$.
            \item For $\neg \Kv_j(\chi,S^k c) \in \Sigma$: If $\mathcal{M}, \vec{h} \models \neg \Kv_j(\chi,S^k c)$ and $\dep(\neg \Kv_j(\chi,S^k c)) \le D - k$, select two witnesses $\vec{h}_1,\vec{h}_2 \in H_{k+1}$ with $\mathcal{M}, \vec{h}_1 \models \chi$ and $\mathcal{M}, \vec{h}_2 \models \chi$ such that $V_\mathbf{C}^\star(c, \vec{h}_1) \neq V_\mathbf{C}^\star(c, \vec{h}_2)$. 
            Add them to $H'_{k+1}$.
            (This can be done because of the construction of $W'$.)
        \end{itemize}
    \end{enumerate}
    Let $W_s = \bigcup_{k=0}^D H'_k$. $W_s$ is finite. Let $\sim_i^s$ be the restriction of $\sim_i'$ to $W_s$, and $V_\mathbf{C}^s$ be the restriction of $V_\mathbf{C}^\star$ to $W_s$. Let $\mathcal{M}_s = \langle W_s, \mathbb{N}^\star, \mathbf{S}^\star, \{\sim_i^s\}_{i\in I}, V_\mathbf{C}^s\rangle$.
We next prove the following lemma.
\begin{lemma}
    \label{lemm:preserve}
     For each $d$ such that $0 \le d \le D$ and for each formula $\psi \in \Sigma$ with $\dep(\psi) = d$, and for any node $\vec{h} \in W_s$, if the height condition
    \( \hgt(\vec{h}) \le D - d \)
    is satisfied, then: $\mathcal{M}_s, \vec{h} \models \psi \iff \mathcal{M}, \vec{h} \models \psi$.
\end{lemma}
\begin{proof}[Proof of the lemma]
By induction on the modal depth $d$ ($0 \le d \le D$).

    \noindent\textbf{Base Case ($d=0$):}
    $\psi$ is atomic or Boolean combination of atomics. The condition $\hgt(\vec{h}) \le D$ always holds for $\vec{h} \in W_s$. Since valuations are inherited directly from $\mathcal{M}$, the equivalence holds trivially.

    \noindent\textbf{Inductive Hypothesis:}
    Assume the Claim holds for all formulas with depth $< d$, i.e., for any formula $\chi\in \Sigma$ with $\dep(\chi) < d$, and for any node $\vec{h} \in W_s$ satisfying $\hgt(\vec{h}) \le D - \dep(\chi)$, we have:
    \( \mathcal{M}_s, \vec{h} \models \chi \iff \mathcal{M}, \vec{h} \models \chi. \)

    \noindent\textbf{Inductive Step:}
    We consider formulas of depth $d$.
    
    \noindent\textit{Subcase 1: $\psi := \K_j \chi$.} Here $\dep(\psi) = d$, so $\dep(\chi) = d-1$.
    Assume the condition holds for $\vec{h}$: $\hgt(\vec{h}) \le D - d$.
    
    ($\Leftarrow$) Suppose $\mathcal{M}, \vec{h} \models \neg\K_j \chi$.
    \begin{itemize}
        \item If $j \neq \Ag(\vec{h})$:
        Since $\hgt(\vec{h}) \le D - d$, the selection condition for $\vec{h}$ is satisfied. Thus, during the construction of $H'_{\hgt(\vec{h})+1}$, a witness $\vec{h}' \in H'_{\hgt(\vec{h})+1}$ was selected such that $\vec{h} \sim'_j \vec{h}'$ and $\mathcal{M}, \vec{h}' \nvDash \chi$.
        As $\hgt(\vec{h}') = \hgt(\vec{h}) + 1 \le D - d + 1 = D - \dep(\chi)$, the condition for IH on $\vec{h}'$ holds. By IH, $\mathcal{M}_s, \vec{h}' \nvDash \chi$, and since $\vec{h} \sim^s_j \vec{h}'$, we have $\mathcal{M}_s, \vec{h} \models \neg\K_j \chi$.
        
        \item If $j = \Ag(\vec{h})$:
        By the property of $W'$, $\vec{h}$ is a $j$-successor of its parent $\vec{p}$. Since $\vec{h} \sim'_j \vec{p}$, $\mathcal{M}, \vec{p} \models \neg\K_j \chi$.
        Since $\hgt(\vec{p}) = \hgt(\vec{h}) - 1 < D - d$, the selection condition for $\vec{p}$ strictly holds. Thus, when constructing successors for $\vec{p}$, a witness $\vec{w} \in H'_{\hgt(\vec{p})+1}$ was selected such that $\vec{p} \sim'_j \vec{w}$ and $\mathcal{M}, \vec{w} \nvDash \chi$.
        In $\mathcal{M}_s$, $\vec{h}, \vec{p}, \vec{w}$ are in the same $j$-cluster, so $\vec{h} \sim^s_j \vec{w}$.
        Moreover, $\hgt(\vec{w}) = \hgt(\vec{p}) + 1 = \hgt(\vec{h}) \le D - d < D - \dep(\chi)$. The condition for IH holds, entailing $\mathcal{M}_s, \vec{w} \nvDash \chi$. Thus $\mathcal{M}_s, \vec{h} \models \neg\K_j \chi$.
    \end{itemize}

    ($\Rightarrow$) Suppose $\mathcal{M}_s, \vec{h} \models \neg\K_j \chi$. There exists $\vec{z} \in W_s$ with $\vec{h} \sim^s_j \vec{z}$ and $\mathcal{M}_s, \vec{z} \nvDash \chi$.
    Since $\vec{z}$ is in the same $j$-cluster as $\vec{h}$ in $W'$, $\vec{z}$ is either $\vec{h}$ itself, its parent, its child, or its sibling. In all cases, $\hgt(\vec{z}) \le \hgt(\vec{h})+1 \le D - \dep(\chi)$. The condition for IH on $\vec{z}$ holds, so $\mathcal{M}_s, \vec{z} \nvDash \chi \implies \mathcal{M}, \vec{z} \nvDash \chi$.
    Since $\vec{h} \sim'_j \vec{z}$ in $\mathcal{M}$, we get $\mathcal{M}, \vec{h} \models \neg\K_j \chi$.

    \noindent\textit{Subcase 2: $\psi := \neg \Kv_j(\chi,S^k c) $.} Here $\dep(\psi) = d$, so $\dep(\chi) = d-1$.
    Assume $\hgt(\vec{h}) \le D - d$.
    ($\Leftarrow$) Suppose $\mathcal{M}, \vec{h} \models \neg \Kv_j(\chi,S^k c)$.
    \begin{itemize}
        \item If $j \neq \Ag(\vec{h})$:
        Since $\hgt(\vec{h}) \le D - d$, the selection condition is satisfied, so there exist two selected witnesses $\vec{h}_1,\vec{h}_2 \in H'_{\hgt(\vec{h})+1}$ with $\vec{h} \sim'_j \vec{h}_1$, $\vec{h} \sim'_j \vec{h}_2$, 
        $\mathcal{M}, \vec{h}_1\models \chi$, $\mathcal{M}, \vec{h}_2 \models \chi$, and $V_\mathbf{C}^\star(c,\vec{h}_1)\neq V_\mathbf{C}^\star(c,\vec{h}_2)$.
        By the IH (applicable since $\hgt(\vec{h}_\ell) \le D - \dep(\chi),\ell=1,2$), $\mathcal{M}_s, \vec{h}_1\models \chi$ and $\mathcal{M}_s, \vec{h}_2 \models \chi$.
        Since $\vec{h} \sim^s_j \vec{h}_1\sim^s_j \vec{h}_2$, we get $\mathcal{M}_s, \vec{h} \models \neg \Kv_j(\chi,S^k c)$.

        \item If $j = \Ag(\vec{h})$:
        Let $\vec{p}$ be the parent of $\vec{h}$. Then $\mathcal{M}, \vec{p} \models \neg \Kv_j(\chi,S^k c)$ and
        $\hgt(\vec{p})=\hgt(\vec{h})-1 < D-d$, so the selection condition for $\vec{p}$ is satisfied.
        Hence there exist two selected witnesses $\vec{w}_1,\vec{w}_2 \in H'_{\hgt(\vec{p})+1}$ in the same $j$-cluster with
        $\mathcal{M}, \vec{w}_1\models \chi$, $\mathcal{M}, \vec{w}_2 \models \chi$, and $V_\mathbf{C}^\star(c,\vec{w}_1)\neq V_\mathbf{C}^\star(c,\vec{w}_2)$.
        By IH, $\mathcal{M}_s, \vec{w}_1\models \chi$ and $\mathcal{M}_s, \vec{w}_2 \models \chi$.
        Since $\vec{h}$ is in the same cluster as $\vec{p}, \vec{w}_1, \vec{w}_2$, $\mathcal{M}_s, \vec{h} \models \neg \Kv_j(\chi,S^k c)$.
    \end{itemize}

    ($\Rightarrow$) Assume $\mathcal{M}_s, \vec{h} \models \neg \Kv_j(\chi,S^k c)$. By semantics, there exist
    $\vec{z}_1,\vec{z}_2 \in W_s$ such that $\vec{h}\sim^s_j\vec{z}_1$, $\vec{h}\sim^s_j \vec{z}_2$, $\mathcal{M}_s, \vec{z}_1\models \chi$, $\mathcal{M}_s, \vec{z}_2 \models \chi$, and $V_\mathbf{C}^s(c,\vec{z}_1)\neq V_\mathbf{C}^s(c,\vec{z}_2)$.
    Since $\vec{z}_1,\vec{z}_2$ are in the same $j$-cluster as $\vec{h}$, $\hgt(\vec{z}_\ell)\le \hgt(\vec{h})+1 \le D-\dep(\chi)$, the IH gives $\mathcal{M}, \vec{z}_\ell \models \chi$ ($\ell=1,2$).
    Moreover, $V_\mathbf{C}^\star(c,\vec{z}_1) = V_\mathbf{C}^s(c,\vec{z}_1) \neq V_\mathbf{C}^s(c,\vec{z}_2) = V_\mathbf{C}^\star(c,\vec{z}_2)$. Since $\vec{h} \sim'_j \vec{z}_1$ and $\vec{h} \sim'_j \vec{z}_2$ in $\mathcal{M}$, we get $\mathcal{M}, \vec{h} \models \neg \Kv_j(\chi,S^k c)$.
   
    \noindent\textit{Subcase 3: boolean connectives.} Trivial.
  \end{proof}  
    \noindent\textbf{Conclusion:}
    Since $\dep(\phi) = D$, and $\hgt(\langle w_0\rangle) = 0$, the condition $0 \le D - D$ holds. Applying the Claim with $d=D$, we get $\mathcal{M}_s, \langle w_0\rangle \models \phi$. Moreover, $|W_s| \le (2|\Sigma|)^{D+1}$.
\end{proof}
This yields the following theorem and corollary.
\begin{theorem}
    $\sys$ is sound and weakly complete w.r.t. standard models.
\end{theorem}
\begin{proof}
    Given a consistent formula $\phi$, it is satisfiable in a non-standard model. Then by Theorem \ref{thm:fmp}, it is satisfiable in a finite non-standard model. By Lemma \ref{lemm:standard}, it is satisfiable in a finite standard model.
\end{proof}

\begin{corollary}
    $\logic$ has finite model property w.r.t. standard models.
\end{corollary}

\medskip We now establish the decidability of $\logic$ using the results above.
\begin{theorem}
    \label{thm:decide}
    $\logic$ is decidable.
\end{theorem}
\begin{proof}
    Given a satisfiable formula $\phi$, it is satisfiable in a finite standard model $\mathcal{M}= \langle W, \mathbb{N}, \mathbf{S}, \{\sim_i\}_{i\in I}, V_\mathbf{C}\rangle$
    of size $f(|\phi|)$, where $f$ is a computable function. Let $C_\phi$ be the set of constants occurring in $\phi$. Let $B= \max(\{a+b\mid \exists  x,y\in \mathbf{C}\cup\{0\},\ S^a x \approx S^b y  \text{ is a subformula of }\phi\}\cup \{0\})$.
    Let $U=\{V_\mathbf{C}(c,w)\mid c \in C_\phi, w\in W\}\cup\{0\}$. We have $|U| \le |C_\phi|\cdot f(|\phi|)+1$. Construct an injective map $g:U\to \mathbb{N}$ as follows: 
    Start from $g(0)=0$, consider the least number $x\in U$ such that $g(x)$ is not defined. Suppose the greatest number $y\in U$ with $g(y)$ defined is $g(y)=m$. If $x-y \le B$, let $g(x)=m+(x-y)$; otherwise, let $g(x)=m+B+1$. 
    Let $\mathcal{M}'= \langle W, \mathbb{N}, \mathbf{S}, \{\sim_i\}_{i\in I}, V'_\mathbf{C}\rangle$ be the new model defined by $V'_\mathbf{C}(c,w)=g(V_\mathbf{C}(c,w))$ for each $c\in C_\phi$ and $w\in W$. By a similar argument as in the proof of Lemma \ref{lemm:standard}, we can show that $\mathcal{M}', w \models \phi$. Hence if $\phi$ is satisfiable, it is satisfiable in a finite standard model, where the size of the model is bounded by $f(|\phi|)$ and $V_\mathbf{C}(c,w)\in \{0,1,\cdots, (|C_\phi|\cdot f(|\phi|) +1)\cdot (B+1)\}$ for each $c\in C_\phi$ and $w\in W$. We can then enumerate all such models 
    (the size of the model is bounded, the codomain of $V_\mathbf{C}$ is bounded, and we only need to consider
         agents occurring in $\phi$ and constants in $C_\phi$) and check if $\phi$ is satisfiable in any of them. 
    There are only finitely many such models, and checking satisfiability in each model is decidable, so we can decide the satisfiability of $\phi$.
    Hence $\logic$ is decidable.
\end{proof}

\section{Consecutive Numbers}
\label{sec:consecutive-numbers}
In this section, we return to the original motivation of this paper, the ``consecutive numbers'' puzzle in Example~\ref{ex:consecutive-puzzle}.
Before presenting the formalization, let us first spell out the intuitive solution of the puzzle. If Anne's number were $0$, Bill's number would have to be $1$, so Anne would know it. Thus Anne's first announcement says, in effect, that her number is not $0$. Similarly, after this announcement, Bill's ignorance says that his number is neither $0$ nor $1$: if it were $0$, Anne's number would be $1$, and if it were $1$, Anne's number could only be $2$.
Anne can now use Bill's announcement. She can know Bill's number only if her own number is $1$ or $2$; for any larger number, Bill could have either the predecessor or the successor. Hence everyone can know the remaining possibilities are precisely $(a,b)=(2,3)$ and $(a,b)=(1,2)$. In either case Bill can determine Anne's number.

The key point is that each public announcement about an agent's epistemic state corresponds to the public announcement of an arithmetic fact about the two numbers. We can capture this process within our system. 
Let $\mathbb{PALKVSA}^r$ be $\sys$ plus !ATOM, !NEG, !CON, !K, $!\text{Kv}^r$, where !ATOM is $[\psi]\alpha \leftrightarrow (\psi \to \alpha)$ for atomic $\alpha$.
!COM is derivable in $\mathbb{PALKVSA}^r$. The following definition formalizes the puzzle and its conclusion.

\begin{definition}
        Let the set of premises be $\Gamma_0=\{\Kv_A a, \Kv_B b, \neg a\approx 0\rightarrow \neg \Kv_A b, \neg b\approx 0\rightarrow \neg \Kv_B a\}\cup \{\K_i(a\approx Sb\vee b\approx Sa) \mid i\in\{A,B,C\}\}$.
        Define $\Gamma_{n+1}:=\Gamma_n\cup\{\K_A\psi, \K_B\psi,\K_C\psi \mid\psi\in \Gamma_n\}$.
    The conclusion is: 
    \[[\neg \Kv_A b][\neg \Kv_B a][\Kv_A b](\Kv_B a\wedge \K_C((a\approx SS0\wedge b\approx SSS0)\vee (b\approx SS 0\wedge a\approx S0)))\]
    \end{definition}
The premise $\neg a\approx 0 \to \neg\Kv_A b$ warrants some explanation. In the puzzle, when Anne's number is 0, she can deduce that Bill's number is 1. Otherwise, she does not know Bill's number. So this premise matches the setup of the puzzle. Although this premise might not be strictly necessary, $\Gamma_0\setminus\{\neg a\approx 0\to \neg Kv_Ab,\neg b\approx 0\to \neg Kv_Ba\}$ is provably insufficient to deduce the conclusion. Moreover, this premise can simplify the proof of the conclusion, so we include it in the premises.
\begin{notation}
        Set $\chi^n_m :=  \bigl(\bigwedge_{k=0}^{m} \neg a\approx S^k 0\bigr)\wedge\bigl(\bigwedge_{k=0}^{n} \neg b\approx S^{k}0\bigr)$, $\varphi_1 := \neg \Kv_Ab$, and
            \begin{align*}
                \varphi_n := \begin{cases}
                \varphi_{n-1} \wedge [\varphi_{n-1}]\neg \Kv_B a & \text{if } n \text{ is even}\\
                \varphi_{n-1} \wedge [\varphi_{n-1}]\neg \Kv_A b & \text{if } n \text{ is odd}
                \end{cases}
                \quad
                \psi_n := \begin{cases}
                [\varphi_{n-1}]\Kv_Ba & \text{if $n$ is even}\\
                [\varphi_{n-1}]\Kv_Ab & \text{if $n$ is odd}    
                \end{cases}
            \end{align*}                                                                             
If $n$ or $m$ is less than $0$, the corresponding part is simply ignored. 
    \end{notation}  
We can solve a generalized version of the puzzle. The following proposition makes precise the intuition that each announcement in the puzzle conveys arithmetic information about the two numbers. 
The proofs are omitted, but they can be established by induction on $n$.
\begin{proposition}
    \begin{enumerate}
\item If $n$ is even, 
	$\Gamma_{n-1}\vdash (\varphi_n\leftrightarrow \chi_{n-2}^{n-1})\wedge (\psi_n\leftrightarrow \neg\chi_{n-2}^{n-1})$
	\item If $n$ is odd, $\Gamma_{n-1}\vdash (\varphi_n\leftrightarrow \chi_{n-1}^{n-2})\wedge (\psi_n\leftrightarrow \neg\chi_{n-1}^{n-2})$.
\end{enumerate}
\end{proposition}
\medskip
With !COM, this yields the following solution in $\mathbb{PALKVSA}^r$.
    \begin{proposition}
	\begin{enumerate}
		\item If $n$ is even, $\Gamma_n\vdash \underbrace{[\neg Kv_A b][\neg Kv_B a]\cdots [\neg Kv_B a]}_{n} [Kv_Ab](Kv_B a\wedge K_C((a\approx S^n0\wedge b\approx S^{n+1}0)\vee (b\approx S^n 0\wedge a\approx S^{n-1}0)))$. 
		\item If $n$ is odd, $\Gamma_n\vdash \underbrace{[\neg Kv_A b][\neg Kv_B a]\cdots [\neg Kv_A b]}_{n} [Kv_Ba](Kv_A b\wedge K_C((b\approx S^n0\wedge a\approx S^{n+1}0)\vee (a\approx S^n 0\wedge b\approx S^{n-1}0)))$.
	\end{enumerate}
\end{proposition}

\section{Conclusion and Future Work}
\label{sec:conclusion}
In this paper, we introduced equality and the successor function symbol into knowing-value logic, resulting in the logic $\logic$. We provided an axiom system $\sys$, which we proved to be strongly complete with respect to non-standard models and weakly complete with respect to standard models. By formalizing and solving the ``Consecutive Numbers'' puzzle, we demonstrated the enhanced expressiveness of our logic. Furthermore, we established the decidability of $\logic$, illustrating that our framework achieves a balance between expressiveness and computational complexity. 

For future work, the language can be further enriched with addition and multiplication symbols, which would enable the formalization and resolution of more complex problems, such as the well-known ``Sum and Product'' puzzle. 
In our current logic, $\Kv_i(\phi,S^kc)$ is equivalent to $\Kv_i(\phi,c)$, which significantly simplifies the completeness proof. However, when extending the language with addition and multiplication, an agent may know the sum or product of two numbers without knowing the exact value of the numbers themselves. Furthermore, addition and multiplication give rise to more complex arithmetic structures.
Consequently, establishing completeness and decidability for such an extended logic will be more challenging. 
Overall, given the rich interplay between mathematical structures and epistemic reasoning, the integration of epistemic logic and arithmetic holds great potential for future research.

\section*{Acknowledgements}
I am deeply grateful to Yanjing Wang for suggesting the topic of this paper and for his encouragement and valuable advice. I thank Yifeng Ding for helpful comments during a seminar presentation. I also thank the three anonymous reviewers from AiML 2026 for their careful reading and insightful suggestions.

\newpage

\bibliographystyle{eptcs}
\bibliography{generic}
\end{document}